\documentclass{lmcs} 

\usepackage{ogonek}
\usepackage{amsmath}
\usepackage{amssymb}
\usepackage{hyperref}

\newcommand{\N}{\mathbb{N}}
\newcommand{\Z}{\mathbb{Z}}
\newcommand{\R}{\mathbb{R}}
\newcommand{\Q}{\mathbb{Q}}

\newcommand{\Qplus}{\Q^+}
\newcommand{\divides}{\mid}
\newcommand{\union}{\cup}
\newcommand{\disjunion}{\dot{\cup}}
\newcommand{\absval}[1]{\lvert #1\rvert}
\newcommand{\onespace}{\quad}
\newcommand{\Th}{\textrm{th}}

\DeclareMathOperator{\diam}{diam}
\DeclareMathOperator{\interior}{int} 

\DeclareMathOperator{\card}{card}
\DeclareMathOperator{\dom}{dom} 
\DeclareMathOperator{\img}{im} 
\DeclareMathOperator{\id}{id} 

\DeclareMathOperator{\ev}{ev} 
\DeclareMathOperator{\Cvx}{Cvx} 
\DeclareMathOperator{\LF}{LF} 

\newcommand{\setconstr}{\divides}
\newcommand{\cl}{\overline}
\newcommand{\compose}{\circ}

\newcommand{\trfns}{R^{(1)}} 
\newcommand{\Bairespc}{\mathbb{B}} 
\newcommand{\Tzero}{\textbf{T}_0} 
\newcommand{\pr}{\ensuremath{\pi}}
\newcommand{\charf}[1]{\chi_{#1}} 

\newcommand{\pwrset}{\mathcal{P}}

\newcommand{\mathst}{\text{ s.t.~}} 
\newcommand{\ball}[2]{B(#1;#2)}
\newcommand{\clball}[2]{\bar{B}(#1;#2)}
\newcommand{\ballmetric}[3]{B_{#1}(#2;#3)}
\newcommand{\clballmetric}[3]{\bar{B}_{#1}(#2;#3)}
\newcommand{\nbhd}[2]{N_{#2}(#1)} 
\newcommand{\bdry}{\partial}

\newcommand{\Cont}{\mathit{C}}  
\newcommand{\leastmu}[2]{\mathbf{\mu}#1\left(#2\right)} 
\newcommand{\ordomega}{\mathit{\omega}} 

\newcommand{\FS}{\textrm{FS}} 

\newcommand{\crsarr}{\rightrightarrows} 
\newcommand{\deltac}{\delta_{\textrm{cover}}} 
\newcommand{\deltamc}{\delta_{\textrm{min-cover}}} 
\newcommand{\deltar}{\delta_{\textrm{range}}} 
\newcommand{\deltacz}{\delta_{\textrm{c}\mathcal{Z}}} 
\newcommand{\leqW}{\leq_{\textrm{W}}} 
\newcommand{\leqsW}{\leq_{\textrm{sW}}} 
\newcommand{\eqvW}{\equiv_{\textrm{W}}} 
\newcommand{\restrict}{\upharpoonright} 
\newcommand{\deltad}{\delta_{\textrm{Dugundji}}}

\theoremstyle{plain} 

\begin{document}
\title[Dugundji systems and retracts]{Dugundji systems and a retract characterization of effective
zero-dimensionality\rsuper*} 
\titlecomment{{\lsuper*}Some results of this paper have been presented at CCA 2015, Tokyo, Japan, Jul 12--15,
in the talk ``Effective zero-dimensionality and retracts''.  Several typos in the abstract for that talk are
corrected here.}
\author[R.~Kenny]{Robert Kenny}	
\address{Department of Mathematics \& Statistics,
School of Physics, Mathematics and Computing,
The University of Western Australia, Perth, Australia} 
\email{robert.kenny@uwa.edu.au}  

%





\begin{abstract}
\noindent In this paper (as in \cite{zerodpaper}), we consider an effective
version of the characterization of separable metric spaces as
zero-dimensional iff every nonempty closed subset is a retract of the space
(actually, it is a relative result for closed zero-dimensional subspaces of
a fixed space that we have proved).  This uses (in the converse direction)
local compactness \& bilocated sets as in \cite{zerodpaper}, but in the
forward direction the newer version has a simpler proof and no compactness
assumption.  Furthermore, the proof of the forward implication relates to
so-called Dugundji systems:
we elaborate both a general construction of such systems for a proper nonempty closed subspace (using a computable form of countable paracompactness), and modifications --- to make the sets pairwise disjoint
if the subspace is zero-dimensional, or to avoid the restriction to proper subspaces.  In a different direction, a second theorem applies in $p$-adic analysis the ideas of the first theorem to compute a more general form of retraction, given a Dugundji system (possibly without disjointness).

Finally, we complement the effective retract characterization of
zero-dimensional subspaces mentioned above
by improving to equivalence the implications (or
Weihrauch reductions in some cases), for closed at-most-zero-dimensional
subsets with some negative information, among separate conditions of computability of operations $N,M,B,S$ introduced in \cite[\S 4]{zerodpaper} and corresponding to
vanishing large inductive dimension, vanishing small inductive
dimension, existence of a countable basis of relatively clopen sets,
and the reduction principle for sequences of open sets.  Thus,
similarly to the robust notion of effective zero-dimensionality of
computable metric spaces in \cite{zerodpaper}, there is a robust
notion of `uniform effective zero-dimensionality' for a represented
pointclass consisting of at-most-zero-dimensional closed subsets.
\end{abstract}

\maketitle
\section{Introduction}\label{sec:intro}
In this paper, we again (after \cite{zerodpaper}) consider
properties of zero-dimensional sets in computable metric spaces, in
the first place related to retractions.  Recall a \emph{retraction} of
a topological space $X$ onto a subset $A$ is a continuous map
$f:X\to X$ such that $A=f(X)$ and $f|_A=\id_A$.  A subset $A$ of $X$ is a
\emph{retract} of $X$ if there exists a retraction $f$ of $X$ onto $A$
\cite[Ex 1.5.C]{Engelking}.  It is well-known that for a nonempty separable
metrizable space $X$, zero-dimensionality is equivalent to the statement
that every closed nonempty subset is a retract of $X$
\cite[Thm 7.3]{Kechris}.  In the framework of computable analysis via representations
\cite{BHW08}, \cite{Weihrauchbook2} (see also \cite{BrattkaPresser},
and \cite{Pauly16} for a more modern treatment), our first theorem
(Theorem \ref{thm:a}) is a computable version of one direction of this result (and a generalisation of \cite[Thm 7.6]{zerodpaper}): assuming certain effective information on the
zero-dimensionality of a closed subset $A$ of a computable metric space $X$, we can compute (via an operation $E'$) retractions of $A$
onto $B$ for any nonempty closed subset $B\subseteq A$
(in this paper without a compactness assumption).

Together with a converse Proposition \ref{prn:aconverse} (generalizing
\cite[Prop 8.5]{zerodpaper}), this will establish computability of $E'$ is
equivalent to computability of any of the four operations of \cite[\S 4]{zerodpaper}, under certain (computable local compactness) assumptions.
More precisely, each of these five operations is defined in relation to a
represented space $\mathcal{Y}$ of subsets of $X$, in this paper consisting
of closed subsets with $\delta_{\mathcal{Y}}\leq\delta_{\Pi^0_1(X)}$; the
implications in \cite[\S 4]{zerodpaper} can then (see
Remark \ref{rem:unifzd} and, using Weihrauch reducibility
\cite{BGPWCpxtyPreprint}, Section \ref{sec:zdsubsets}) be improved to a robust notion of `uniform effective zero-dimensionality' for
such $\mathcal{Y}$ (with computability of $E'$ also equivalent under
stronger assumptions).

In a different direction we also establish, in Theorem \ref{thm:b}, computability (given some additional data) of a
retraction onto a closed homeomorph $A\subseteq X$ of an ultrametric ball in $K$, when
$X$ is an effectively zero-dimensional computable metric space and $K=\Omega_p$ is the nonarchimedean valued field
of $p$-adic numbers (here we refer to \cite{Schikhof}; for our
concerns, also \cite{KKSNonarchDET}, \cite{KapoulasThesis} may be
relevant, though not used in this paper).
As an existence statement, this may seem less general than
Theorem \ref{thm:a}, however the idea here
is to consider a more general form of retraction (onto a less general closed set $A$); the `additional data'
required for defining \& computing the retraction consists of: a Dugundji system
$((V_i)_{i\in\N},(y_i)_{i\in\N})\in\Sigma^0_1(X)^{\N}\times X^{\N}$ for
$A$ (see below and in Section \ref{sec:thirdagain}),
a $K$-valued analogue of a partition of unity subordinate to
$(V_i)_{i\in\N}$, and $h:A\to h(A)\subseteq K$
the homeomorphism mentioned above.  Though well-known, it is worth
mentioning any ultrametric space can (at least
noneffectively) be embedded into some nonarchimedean valued field $K$; we
record this fact as Theorem \ref{thm:zdembcor} and will not enquire about
the computable content of that result for the purposes of this paper.

A secondary aim\footnote{Note that Dugundji systems are used in the
construction of retractions in Theorems \ref{thm:a} and \ref{thm:b}.} of
the paper is to elaborate certain aspects of the construction of a Dugundji
system for a closed subset $\emptyset\neq B\neq X$ of a computable metric
space $X$.  For us, a Dugundji system for $B$
consists of a (locally finite in its union) collection $(V_i)_{i\in\N}\subseteq\Sigma^0_1(X)$ and
a sequence of points $(y_i)_{i\in\N}\subseteq B$ satisfying certain axioms about the relative closeness of $V_i$
to $y_i$ and to $B$ (for the precise definition we refer to Section \ref{sec:thirdagain}).  These
objects are used in the proof of the Dugundji Extension Theorem
(cf.~\cite[Hint to Prob 4.5.20(a)]{Engelking}, \cite[\S 1.2]{vanMill2} or,
originally, \cite{Dugundji}), and could be said (informally speaking) to
have some attributes of piecewise constant maps $X\setminus B\to X$.
As such it is interesting to examine representations of the class of closed sets $\mathcal{A}(X)$ related to Dugundji systems
(in this paper, see Proposition \ref{prn:dclosedsets}).  However, our main interest is to give a general construction of a Dugundji system (Proposition \ref{prn:dugundji}), from
$\deltar\sqcap\delta^{>}_{\textrm{dist}}$-information on $B$.  This construction is developed in several
respects further than we need to prove the above two theorems --- for instance, we improve the coefficient $2$
in the bound on $d(x,y_i)$ ($x\in V_i$) to $1+\epsilon$, and replace the open balls $\ball{x}{2^{-2}d_A(x)}$
($x\in X\setminus B$) by $\ball{x}{f(x).d_A(x)}$ for a specified lower semicontinuous $f:X\setminus B\to (0,1]$
with a strict upper bound $\frac{\epsilon}{2+\epsilon}$ specified along a countable dense subset
$(\img\nu)\setminus B$.  Such a function $f$ is easily found for any $\epsilon$ and $B$, but the more general
conditions may be useful in some circumstances.  We also note the construction uses, initially,
an effective countable paracompactness result (Lemma \ref{lem:effparac}) which may be of independent interest.


Next, in Section \ref{sec:newfourth}, from such a construction, under a certain effective assumption of
zero-dimensionality for a closed set $A$, Dugundji systems relative to $A$ for any closed $B\subseteq A$ can be
computed (Proposition \ref{prn:dugundjizd}) with pairwise disjoint (open) sets, provided $\emptyset\neq B\neq X$.
In case that $X$ is an effectively zero-dimensional computable metric space, a slightly different construction
is presented (Remark \ref{rem:zddgone}), which allows to obtain clopen sets here, as well as (e.g.) to avoid the
restriction $B\neq X$ under an additional assumption (Remark \ref{rem:zddgtwo}).

Finally, Section \ref{sec:fourth} examines an application of convexity in nonarchimedean valued fields to define
more general retractions and in particular to establish Theorem \ref{thm:b}, Section \ref{sec:zdsubsets}
as mentioned establishes the validity of the notion of `uniform effective zero-dimensionality' referred to above,
and Section \ref{sec:concl} outlines some 
directions for future work.

%

\section{Notation}\label{sec:second}
Besides the results and concepts found in the references
already mentioned for computable analysis via representations
\cite{BHW08}, \cite{Weihrauchbook2}, \cite{BrattkaPresser} (also
\cite{Pauly16}), we also need occasionally to refer to some sections
of our earlier paper on zero-dimensionality \cite{zerodpaper}, and for
consistency's sake --- together with our use of
$\N$ and $\Bairespc:=\N^{\N}$ for names --- this determines the use of
(e.g.) symbols $\delta_{\Sigma^0_1(X)},\delta_{\Pi^0_1(X)},\deltac$ for
common representations in place of
$\theta^{\textrm{en}}_{<},\psi^{\textrm{en}}_{>},\kappa_{\textrm{c}}$
in \cite{Weihrauchbook2}.

For general topology concepts we refer to \cite{Engelking} (some other references useful for dimension theory,
including dimension zero, are \cite{vanMill2}, \cite{Kuratowski} and \cite[App I]{Lipscomb}), and for
nonarchimedean analysis to \cite{Schikhof}.
In the remainder of this section we will recall some definitions as well as
basic results on partial continuous maps useful in the rest of the paper.

For any $p\in\Bairespc$ and $i\in\N$ we denote $p_i:=p(i)\in\N$,
while by $(p^{(i)})_{i\in\N}$ we usually mean an arbitrary sequence of
elements of $\Bairespc$.  Bijection
\(\langle\cdot,\dots\rangle:\Bairespc^{\N}\to\Bairespc,(p^{(i)})_{i\in\N}
\mapsto q\) is defined by $q_{\langle j,k\rangle}:=p^{(j)}_k$ ($j,k\in\N$);
notwithstanding this usage, usually the symbols $p, q, p^{(i)}$ etc.~are
meant to be independent members of $\Bairespc$.  On the other hand, a fixed
computable bijection $\langle\cdot,\cdot\rangle:\N^2\to\N$ with computable
projections $\pi_i:\N\to\N,\langle a_1,a_2\rangle\mapsto a_i$ ($i=1,2$)
will be used throughout the paper.

Now let $(X,\mathcal{T})$ be a second countable topological space, and let
$\alpha,\beta:\N\to\mathcal{T}$ be total numberings of possibly different
countable bases (basis numberings are assumed to be total throughout).
\begin{defi}\label{def:formincl}
$(\sqsubset) \subseteq \N^2$ is a \emph{formal inclusion} of $\alpha$ with respect to $\beta$ if
$(\forall a,b\in\N)(a\sqsubset b\implies \alpha(a)\subseteq\beta(b))$.  We consider the following axioms, in order of increasing strength.
\begin{enumerate}
\item\label{inclone}
 $(\forall b)(\forall x\in X)(\exists a)(x\in\beta(b)\implies x\in\alpha(a)\wedge a\sqsubset b)$
\item\label{incltwo}
 $(\forall b)(\forall x\in X)(\forall U\in\mathcal{T})(\exists a)\left(
 x\in\beta(b)\cap U\implies x\in\alpha(a)\subseteq U\wedge a\sqsubset b \right)$
\item\label{inclthree}
 $(\forall a,b)(\forall x\in X)(\exists c)(x\in\beta(a)\cap\beta(b)\implies
 x\in\alpha(c)\wedge c\sqsubset a\wedge c\sqsubset b)$
\end{enumerate}
Thus if $(X,\mathcal{T})$ is a $\Tzero$ space with basis
numbering $\alpha:\N\to\img\alpha\subseteq\mathcal{T}$, we note
$(X,\mathcal{T},\alpha)$ is a computable topological space iff there exists
a c.e.~formal inclusion $\sqsubset$ of $\alpha$ with respect to $\alpha$
satisfying (\ref{inclthree}).  A formal inclusion will be called a
\emph{refined inclusion} if (\ref{inclone}) holds.
\end{defi}
For two representations $\delta,\delta'$ we write $\delta\leq\delta'$ if there exists a computable map $F:\subseteq\Bairespc\to\Bairespc$
such that $\dom\delta\subseteq F^{-1}\dom\delta'$ and
$(\forall p\in\dom\delta)\delta(p)=(\delta'\compose F)(p)$.  Also, for a
$\Tzero$ space $(X,\mathcal{T}_X)$ with total basis numbering $\alpha$, we
define the standard unpadded representation $\delta_X$ of $X$ by
\(p\in\delta_X^{-1}\{x\}:\iff \img p=\{a\in\N\setconstr \alpha(a)\ni x\}\).

Next, recall a \emph{computable metric space} $(X,d,\nu)$ consists of a
nonempty separable metric space $(X,d)$ and a
dense total sequence $\nu:\N\to X$ such that
$d\compose(\nu\times\nu):\N\times\N\to\R$ is computable with respect
to standard representations.  In this context,
the \emph{Cauchy representation} $\rho_{\nu}$ will be defined by
\[
p\in\rho_{\nu}^{-1}\{x\}:\iff
(\forall i,j)d(\nu(p_i),\nu(p_j))<2^{-\min\{i,j\}}\wedge
\lim_{i\to\infty}\nu(p_i)=x
\]
(the limit being required to exist in $(X,d)$),
while $\nu_{\Qplus}$ is a standard total numbering of $\Qplus=
\{q\in\Q\setconstr q>0\}$ with a $(\nu_{\Qplus},\id_{\N})$-computable
right-inverse $\overline{\cdot}:\Qplus\to\N$.  In any computable metric
space $(X,d,\nu)$, we consider numberings of ideal open and closed
balls
\begin{align*}
&\alpha:\N\to\img\alpha\subseteq\mathcal{T},\langle a,r\rangle\mapsto
\ballmetric{d}{\nu(a)}{\nu_{\Qplus}(r)},\onespace\text{ and }\\
&\hat{\alpha}:\N\to\img\hat{\alpha}\subseteq\Pi^0_1(X),
\langle a,r\rangle\mapsto \clballmetric{d}{\nu(a)}{\nu_{\Qplus}(r)}.
\end{align*}
We will call $\alpha$ the \emph{standard ball numbering} and
denote $\mathcal{B}:=\{\ballmetric{d}{\nu(a)}{q}\setconstr
q\in\Qplus\wedge a\in\N\}=\img\alpha$.  The relation $\sqsubset$
defined by $\langle a,r\rangle\sqsubset\langle b,q\rangle :\iff
d(\nu(a),\nu(b))+\nu_{\Qplus}(r)<\nu_{\Qplus}(q)$ is a
formal inclusion of $\alpha$ with respect to itself; moreover it satisfies $c\sqsubset d\implies\hat{\alpha}(c)\subseteq\alpha(d)$ and
(\ref{inclthree}).

From any basis numbering $\alpha$ of a
topological space $X$ we can define a representation
$\delta:\subseteq\Bairespc\to\mathcal{O}(X),p\mapsto \union\{\alpha(p_i-1)\setconstr i\in\N, p_i\geq 1\}$ of the hyperspace of open sets in $X$.
For a computable metric space with
$\alpha$ as above, this representation is denoted $\delta_{\Sigma^0_1(X)}$, or $\delta_{\Sigma^0_1}$ if $X$ is
clear from the context\footnote{Many of the hyperspace representations
used in this paper can be replaced by (hyperspace representation)
definitions which generalise to arbitrary represented spaces
\cite{Pauly16}, but we will not need them here.}.
\begin{lem}\label{lem:formincl}
Let $\alpha$, $\beta$ be basis numberings for topological space
$(X,\mathcal{T})$, inducing representations
$\delta,\delta':\subseteq\Bairespc\to\mathcal{O}(X)$ resp.  If there exists
a
c.e.~refined inclusion $\sqsubset$ of $\alpha$ with respect to $\beta$,
then $\delta\leq\delta'$.
\end{lem}
\begin{proof}
If $X=\emptyset$ then $\mathcal{T}=\{\emptyset\}$ so (by totality)
$\alpha=\beta$, implying $\delta=\delta'$.  If $X\neq\emptyset$, condition
(\ref{inclone}) ensures $\emptyset\neq (\sqsubset)\subseteq\N^2$, so let
$h\in\trfns$ such that
$\img h=\{\langle a,b\rangle\setconstr a\sqsubset b\}$.  Now computable $F:\subseteq\Bairespc\to\Bairespc$
such that
\[
\{F(p)_i-1\setconstr i\in\N, F(p)_i\geq 1\} = \{a\in\N\setconstr (\exists i)p_i\geq 1\wedge a\sqsubset p_i-1\}
\]
can (e.g.) be defined
by
\[
F(p)_{\langle i,n\rangle}=\left( 1+\pr_1 h(n), \textrm{ if $p_i=1+\pr_2 h(n)$};\:\: 0, \textrm{ otherwise}
\right) \text{ ($i,n\in\N$).}
\]
\end{proof}

Consider now a computable metric space $(X,d,\nu)$ and a
subset $Y\subseteq X$.  We denote the induced basis numbering by
$\alpha_Y:\N\to\mathcal{T}_Y,a\mapsto \alpha(a)\cap Y$.  If there exists
computable $\lambda:\N\to X$ having $\img\lambda\subseteq Y\subseteq
\cl{\img\lambda}$, we say $Y$ is \emph{effectively separable}; in this case
$\alpha_Y$ and the standard ball numbering $\beta$ of
$(Y,d|_{Y\times Y},\lambda)$ give two representations $\delta,\delta'$ of
$\mathcal{O}(Y)$.  Moreover there exist c.e.~formal inclusion relations
given by
\begin{align*}
\langle k,l\rangle \sqsubset_0 \langle i,j\rangle &:\iff
d(\lambda(i),\nu(k))+\nu_{\Qplus}(l)<\nu_{\Qplus}(j),\\
\langle k,l\rangle \sqsubset_1 \langle i,j\rangle &:\iff
d(\nu(i),\lambda(k))+\nu_{\Qplus}(l)<\nu_{\Qplus}(j).
\end{align*}
\begin{lem}
For a computable metric space $(X,d,\nu)$ and effectively separable
$Y\subseteq X$ witnessed by $\lambda:\N\to Y$, $\sqsubset_0$ is a refined
inclusion of $\alpha_Y$ with respect to $\beta$ and $\sqsubset_1$ is a
refined inclusion of $\beta$ with respect to $\alpha_Y$.
\end{lem}
\begin{proof}
If $y\in\beta\langle i,j\rangle$ then pick $l$ with
$\nu_{\Qplus}(l)+d(y,\lambda(i))<\nu_{\Qplus}(j)$
and
\[
k\in\nu^{-1} \left(
\ballmetric{d}{y}{\nu_{\Qplus}(l)}\cap
\ballmetric{d}{y}{\nu_{\Qplus}(j)-\nu_{\Qplus}(l)-d(y,\lambda(i))} \right)
\]
(to get $y\in\alpha_Y\langle k,l\rangle\wedge \langle k,l\rangle\sqsubset_0\langle i,j\rangle)$.  If
$y\in\alpha_Y\langle i,j\rangle$, pick $l$ with
$\nu_{\Qplus}(l)+d(y,\nu(i))<\nu_{\Qplus}(j)$
and
\[
k\in\lambda^{-1} \left(
\ballmetric{d|_{Y\times Y}}{y}{\nu_{\Qplus}(l)}\cap
\ballmetric{d|_{Y\times Y}}{y}{\nu_{\Qplus}(j)-\nu_{\Qplus}(l)-d\left(y,\nu(i)\right)} \right).
\]
\end{proof}
Consequently, we find $\delta\equiv\delta'=\delta_{\Sigma^0_1(Y)}$.

We will need \(\mathcal{G}_{\delta}(X):=\{\bigcap_{i\in\N}U_i\setconstr
(U_i)_i\in\Sigma^0_1(X)^{\N}\}\), the representation
$\delta_{\N}:\Bairespc\to\N,p\mapsto p_0$ of $\N$, a certain canonical
representation $\eta$ of
\begin{equation*}
\mathbf{F}:=\{f:\subseteq\Bairespc\to\Bairespc\setconstr f
\text{ continuous and }\dom f\in\mathcal{G}_{\delta}(\Bairespc)\}
\end{equation*}
(see below) and the following three representations of the class
$\mathcal{A}(X)$ of closed subsets of a computable metric space $X$
(cf.~\cite{BrattkaPresser}), where recall
$\overline{\R}:=\R\union\{-\infty,\infty\}$.
Define $\delta_{\Pi^0_1(X)},\deltar,\delta_{\textrm{dist}}^{>}:\subseteq\Bairespc\to\mathcal{A}(X)$ by
\begin{align*}
p\in\delta_{\Pi^0_1(X)}^{-1}\{A\} &:\iff p\in\delta_{\Sigma^0_1(X)}^{-1}\{X\setminus A\},\\
\langle p^{(0)},\dots\rangle\in\deltar^{-1}\{A\} &:\iff (A=\emptyset\wedge(\forall i)p^{(i)}=0^{\ordomega})\vee
\big( A\neq\emptyset\wedge\{p^{(i)}\setconstr i\in\N\}\subseteq P^{-1}\delta_X^{-1}A\wedge\\
& (\forall x\in A)(\forall U\in\mathcal{T}_X)(\exists i)(x\in U\implies (\delta_X\compose P)(p^{(i)})\in U)
\big),
\end{align*}
where $P:\subseteq\Bairespc\to\Bairespc$ is defined by $P(p)_i:=p_i-1$
($\dom P=\{p\in\Bairespc\setconstr (\forall i)p_i\geq 1\}$),
\[
p\in(\delta_{\textrm{dist}}^{>})^{-1}\{A\} :\iff \eta_p \text{ $(\delta_X,\overline{\rho_{<}})$-realizes }
d_A:X\to\overline{\R},
\]
where
\[
p\in\overline{\rho_{<}}^{-1}\{t\} :\iff \{n\in\N\setconstr \nu_{\Q}(n)<t\} =
\{p_i-1\setconstr i\in\N\wedge p_i\geq 1\}.
\]

We will refer to $\delta_{\Pi^0_1}$ and
$\delta_{\textrm{dist}}^{>}$ as embodying negative information and
$\deltar$ as positive information on closed sets.  Note
$\delta_{\textrm{dist}}^{>}\leq\delta_{\Pi^0_1}$.  We denote the
represented space $(\mathcal{A}(X),\delta_{\Pi^0_1(X)})$ by $\Pi^0_1(X)$,
and will write $\Pi^0_1(X)$ to mean $\mathcal{A}(X)$ when (e.g.) defining
operations.  Note that the operation
\begin{equation*}
t_4:\subseteq\Pi^0_1(X)^2\crsarr\Sigma^0_1(X)^2,(A,B)\mapsto\{(U,V)\setconstr A\subseteq U\wedge B\subseteq V\wedge U\cap V=\emptyset\}
\end{equation*}
($\dom t_4=\{(A,B)\setconstr A\cap B=\emptyset\}$) is computable (see
\cite[Prop 3.5]{zerodpaper}).  In general recall a multi-valued function
$f:\subseteq X\crsarr Y$ is computable if there exists computable
$F:\subseteq\Bairespc\to\Bairespc$ such that
$\delta_X^{-1}\dom f\subseteq F^{-1}\dom\delta_Y$ and \((\forall p\in
\delta_X^{-1}\dom f) (\delta_Y\compose F)(p)\in (f\compose\delta_X)(p)\).

Two representations of the class $\mathcal{K}(X)$ of compact
subsets of $X$ will be used, namely $\deltac$ and $\deltamc$, where
each name $p$ of $K$ is respectively an unpadded list of all codes of
ideal covers (codes of irredundant ideal covers) of $K$.  Here $w\in\N^{*}$ is an ideal cover of $K$ if $\bigcup_{i<\absval{w}}\alpha(w_i)\supseteq K$,
and \emph{irredundant} if
$(\forall i<\absval{w})\alpha(w_i)\cap K\neq\emptyset$.  Thus
e.g.~\(p\in\deltac^{-1}\{K\}\iff \{\langle w\rangle\setconstr w\in\N^{*}
\wedge\bigcup_{i<\absval{w}}\alpha(w_i)\supseteq K\}=\img p\).  We also
denote $\mathcal{K}^{*}(X):=\mathcal{K}(X)\setminus\{\emptyset\}$.

Finally in this section, consider again the class $\mathbf{F}$ and its
canonical representation
$\eta:\Bairespc\to\mathbf{F}$; see \cite[Thm 2.3.13]{Weihrauchbook2} for the closely analogous properties in case of
continuous functions $\subseteq\Sigma^{\ordomega}\to\Sigma^{\ordomega}$ for a finite alphabet $\Sigma$.  We
recall $\eta$ satisfies both the utm (universal Turing machine) and smn properties, that is:
\begin{gather*}
\begin{split}
\text{utm:}\onespace(\exists \text{ computable }u\in\mathbf{F})(\forall p,q\in\Bairespc)\bigl(
 & (q\in\dom\eta_p\iff\langle p,q\rangle\in\dom u) \\
& \qquad \wedge (q\in\dom\eta_p\implies \eta_p(q)=u\langle p,q\rangle) \bigr),
\end{split}\\
\begin{split}
\text{smn:}\onespace & (\forall \textrm{ computable }f\in\mathbf{F})
(\exists \textrm{ computable }S\in\mathbf{F})(\forall p,q\in\Bairespc)\big( S \text{ total }\wedge \\
& \qquad (q\in\dom\eta_{S(p)}\iff\langle p,q\rangle\in\dom f)
\wedge (\langle p,q\rangle\in\dom f\implies f\langle p,q\rangle=\eta_{S(p)}(q)) \big).
\end{split}
\end{gather*}

Some simple applications of these properties
are to be found in dealing with spaces of relatively continuous functions.  For any set $X$, denote by $\pwrset(X)$ the power set of $X$.  For
sets $X,Y$ equipped with
representations $\delta_X,\delta_Y$, and a represented set
$\mathcal{Z}\subseteq\pwrset(X)$, consider the set
$\Cont_{\mathcal{Z}}(\delta_X,\delta_Y)$ and representation
$[\delta_X\to\delta_Y]_{\mathcal{Z}}=\deltacz$
defined by
\begin{align*}
\Cont_{\mathcal{Z}}(\delta_X,\delta_Y) &:=\{f:\subseteq X\to Y\setconstr
f\text{ $(\delta_X,\delta_Y)$-continuous and }\dom f\in\mathcal{Z}\}\\
\langle p,q\rangle\in\deltacz^{-1}\{f\} &:\iff \left( \eta_p
\text{ a $(\delta_X,\delta_Y)$-realizer of } f\right)\wedge \dom f=\delta_{\mathcal{Z}}(q).
\end{align*}
Also denote $\Cont(\delta_X,\delta_Y):=
\Cont_{\{X\}}(\delta_X,\delta_Y)$.  Note in general the definition of
$\deltacz$ depends on choice of $\delta_X,\delta_Y$:
\begin{lem}
Suppose $(X,\delta_X)$, $(Y,\delta_Y)$ are represented spaces,
$\mathcal{Z}\subseteq\pwrset(X)$ with representation
$\delta_{\mathcal{Z}}$.  If $\delta_X',\delta_Y'$ are
representations of $X,Y$ with
$\delta_X'\leq\delta_X$ and $\delta_Y\leq\delta_Y'$ then
corresponding variant representation $\deltacz'$ of
$\Cont_{\mathcal{Z}}(\delta_X',\delta_Y')$ has
$\deltacz\leq\deltacz'$.
\end{lem}
\begin{proof}
If $f\in\Cont_{\mathcal{Z}}(\delta_X,\delta_Y)$ there
exists a continuous $(\delta_X,\delta_Y)$-realizer $K$ of $f$, so if
$F,G:\subseteq\Bairespc\to\Bairespc$ are respective witnesses of
$\delta_X'\leq\delta_X$ and $\delta_Y\leq\delta_Y'$ we claim
$f$ is $(\delta_X',\delta_Y')$-realized by $H=G\compose K\compose F$:
we find
any $p\in(\delta_X')^{-1}\dom f$
has \(p\in F^{-1}\delta_X^{-1}\{\delta_X'(p)\}\subseteq
F^{-1}K^{-1}\dom\delta_Y\subseteq F^{-1}K^{-1}G^{-1}\dom\delta_Y'=
H^{-1}\dom\delta_Y'\) with
\begin{equation*}
(\delta_Y'\compose H)(p)=(\delta_Y\compose K\compose F)(p)=
(f\compose\delta_X\compose F)(p)=(f\compose\delta_X')(p)
\end{equation*}
since $(K\compose F)(p)\in\dom\delta_Y$, $F(p)\in\delta_X^{-1}\dom f$
and $p\in\dom\delta_X'$.  That is,
\((\delta_X')^{-1}\dom f\subseteq H^{-1}\dom\delta_Y'\wedge
(\forall p\in(\delta_X')^{-1}\dom f)\left(
(\delta_Y'\compose H)(p)=(f\compose\delta_X')(p) \right)\) so
$f\in\Cont_{\mathcal{Z}}(\delta_X',\delta_Y')$.  Moreover, if
\[
L:\subseteq\Bairespc\to\Bairespc,\langle p,r\rangle\mapsto
(G\compose\eta_p\compose F)(r)
\]
is considered in the smn property
for $\eta$, then any corresponding (total) $S$ has for $M:\Bairespc\to\Bairespc,\langle p,q\rangle\mapsto\langle S(p),q\rangle$ that $M$
witnesses $\deltacz\leq\deltacz'$: for $f\in\Cont_{\mathcal{Z}}(\delta_X,\delta_Y)$, any $p'=\langle p,q\rangle\in\deltacz^{-1}\{f\}$
has $\pr_2 M(p')=q\in\delta_{\mathcal{Z}}^{-1}\{\dom f\}$ and any
$r\in(\delta_X')^{-1}\dom f\subseteq
(G\compose\eta_p\compose F)^{-1}\dom\delta_Y'$ has
\begin{align*}
r\in\dom(G\compose\eta_p\compose F) &\iff
\langle p,r\rangle\in\dom L \iff r\in\dom\eta_{S(p)} \\
&\iff r\in\dom\eta_{\pr_1 M(p')} \iff \textrm{True},
\end{align*}
with \(\eta_{\pr_1 M(p')}(r)
=\eta_{S(p)}(r)=L\langle p,r\rangle=(G\compose\eta_p\compose F)(r)\)
and
\[
(\delta_Y'\compose \eta_{\pr_1 M(p')})(r)=(\delta_Y'\compose G\compose\eta_p\compose F)(r)=(\delta_Y\compose\eta_p\compose F)(r)=
(f\compose\delta_X\compose F)(r)=(f\compose\delta_X')(r).
\]
So
\((\forall p'\in\dom\deltacz)M(p')\in(\deltacz')^{-1}\{f\}\) and in particular \(\dom\deltacz\subseteq M^{-1}\dom\deltacz'\wedge
(\forall p'\in\dom\deltacz)\deltacz(p')=(\deltacz'\compose M)(p')\).
This completes the proof.
\end{proof}

We will generally write $\Cont_{\mathcal{Z}}(X,Y)$ if the
represented spaces $(X,\delta_X)$, $(Y,\delta_Y)$ are clear.
Where not otherwise specified, if $X,Y$ are second countable $\Tzero$
spaces equipped with basis numberings $\alpha_X,\alpha_Y$ we will
assume the standard unpadded representations of $X,Y$ related to
$\alpha_X$, $\alpha_Y$ are used.
\begin{lem}\label{lem:eval}
Suppose $X,Y$ are $\Tzero$ spaces equipped with representations
$\delta_X,\delta_Y$, and $\mathcal{Z}\subseteq\pwrset(X)$ is a
represented set.
The operator $\ev':\subseteq\Cont_{\mathcal{Z}}(X,Y)\times X\to Y,(f,x)\mapsto f(x)$ with
$\dom\ev'=\{(f,x)\setconstr x\in\dom f\}$ is $([\deltacz,\delta_X],\delta_Y)$-computable.
\end{lem}
\begin{proof}
Let $u:\subseteq\Bairespc\to\Bairespc$ be computable witnessing the utm property for $\eta$.
Then $I:\subseteq\Bairespc\to\Bairespc,\langle\langle p,q\rangle,r\rangle\mapsto u\langle p,r\rangle$ (with
natural domain) is computable and we show any $\langle p',r\rangle \in
[\deltacz,\delta_X]^{-1}\dom\ev'$ has $\langle p',r\rangle\in
I^{-1}\dom\delta_Y$ and $(\delta_Y\compose I)\langle p',r\rangle =
(\ev'\compose[\deltacz,\delta_X])\langle p',r\rangle$.

Namely, if $(f,x)\in\dom\ev'$ and
$\langle\langle p,q\rangle,r\rangle\in[\deltacz,\delta_X]^{-1}\{(f,x)\}$ then
\[
r\in\delta_X^{-1}\{x\}\subseteq \delta_X^{-1}\dom f\subseteq \eta_p^{-1}\dom\delta_Y
\]
with $(\delta_Y\compose\eta_p)(r)=(f\compose\delta_X)(r)$.  So, $\langle p,r\rangle\in\dom u$ with
$u\langle p,r\rangle=\eta_p(r)\in\dom\delta_Y$, but then $\langle\langle p,q\rangle,r\rangle \in
I^{-1}\dom\delta_Y$.  Moreover,
\[
(\ev'\compose[\deltacz,\delta_X])\langle\langle p,q\rangle,r\rangle=(f\compose\delta_X)(r)=
(\delta_Y\compose\eta_p)(r)=(\delta_Y\compose I)\langle\langle p,q\rangle,r\rangle.
\]
This completes the proof.
\end{proof}

\begin{lem}\label{lem:preimg}
Suppose $X,Y$ are $\Tzero$ spaces equipped with basis numberings $\alpha_X,\alpha_Y$ and
c.e.~formal inclusions $\sqsubset$, $\sqsubset'$ satisfying (\ref{inclthree}), (\ref{inclone})
from Definition \ref{def:formincl} respectively.  The operator
\[
v_{\mathcal{Z}}:\Cont_{\mathcal{Z}}(X,Y)\times\Sigma^0_1(Y)\crsarr\Sigma^0_1(X),
(f,U)\mapsto \{V\setconstr V\cap\dom f=f^{-1}U\}
\]
is $([\deltacz,\delta_{\Sigma^0_1(Y)}],\delta_{\Sigma^0_1(X)})$-computable.
\end{lem}
\begin{proof}
We consider dovetailed simulations of a fixed type $2$ TM $\mathcal{M}$ computing $I:\subseteq\Bairespc\to\Bairespc$, the realizer of $\ev'$ from
the proof of Lemma \ref{lem:eval}.  Namely, on input
$\langle\langle p,q\rangle,s\rangle\in\Bairespc$, dovetail output of
$\ordomega$ copies of $0$ with output of $c+1$ over all $a,c\in\N$,
$n\geq 1$, $w\in\N^n\subseteq\N^{*}$ such that
$(\exists i)s_i\geq 1\wedge a\sqsubset' s_i-1$,
$(\forall i<n)(c\sqsubset w_i)$ and such that $a$ appears in the
output of $\mathcal{M}$ on input $\langle\langle p,q\rangle,r\rangle$
without the input head having read past the first $2 n$ places
($r:=w.0^{\ordomega}$).  Plainly this algorithm defines a computable
map $\subseteq\Bairespc\to\Bairespc$.

If $\langle p,q\rangle\in\deltacz^{-1}\{f\}$,
$s\in\delta_{\Sigma^0_1(Y)}^{-1}\{U\}$ and the
entire output is $t\in\Bairespc$, we show
$\delta_{\Sigma^0_1(X)}(t)\cap\dom f=f^{-1}U$.
First, any nonzero output $c+1$ has $\alpha_X(c)\subseteq
\bigcap_{i<n}\alpha_X(w_i)$ for some $a$, $n$, $w\in\N^n$ as above,
but any $x$ in this intersection has a $\delta_X$-name $r$ extending
$r\restrict n=w$.  Then if $x\in\dom f$ we have
\[
f x=(f\compose\delta_X)(r) = (\delta_Y\compose\eta_p)(r)=
(\delta_Y\compose I)\langle\langle p,q\rangle,r\rangle\in\alpha_Y(a).
\]
This shows $f(\bigcap_{i<n}\alpha_X(w_i)\cap\dom f) \subseteq
\alpha_Y(a)\subseteq U$ and in particular
$\alpha_X(c)\cap\dom f\subseteq f^{-1}U$.

Conversely, any $x\in f^{-1}U$ has (by (\ref{inclone}) of Definition \ref{def:formincl}) some $i,a$ such that
$s_i\geq 1\wedge a\sqsubset' s_i-1\wedge\alpha_Y(a)\ni f x$.  For arbitrary
$r\in\delta_X^{-1}\{x\}$ we know
\[
\langle\langle p,q\rangle,r\rangle \in
[\deltacz,\delta_X]^{-1}\dom\ev'\subseteq I^{-1}\dom\delta_Y
\]
and
$(\delta_Y\compose I)\langle\langle p,q\rangle,r\rangle=f x$, so in particular $a$ appears in
$I\langle\langle p,q\rangle,r\rangle$.  Fix $n\geq 1$ such that $a$ is produced by $\mathcal{M}$ on input
$\langle\langle p,q\rangle,r\rangle$ without the input head reading the $(2 n)^{\Th}$ input place or higher;
$w:=r\restrict n$, $\langle p,q\rangle\restrict n$ and $a$ now satisfy the requirements of the above algorithm.  By repeated use
of (\ref{inclthree}) from Definition \ref{def:formincl}, there exists $c$ such that
$x\in\alpha_X(c)\wedge (\forall i<n)(c\sqsubset w_i)$, and then $c+1$ appears in the output of our algorithm.  Since $x$ was arbitrary, this
establishes $f^{-1}U\subseteq \delta_{\Sigma^0_1(X)}(t)$ as desired.
%
%
\end{proof}
\begin{cor}
Under the conditions of the lemma,
\[
a_{\mathcal{Z}}:\Cont_{\mathcal{Z}}(X,Y)\times\Pi^0_1(Y)\crsarr
\Pi^0_1(X), (f,A)\mapsto \{C\setconstr C\cap\dom f=f^{-1}A\}
\]
is computable.
\end{cor}
\begin{proof}
Suppose $f\in\Cont_{\mathcal{Z}}(X,Y)$.  Since $\dom f\setminus f^{-1}A=f^{-1}(Y\setminus A)$ we have
\[
C\in a_{\mathcal{Z}}(f,A)\iff \dom f\setminus C=f^{-1}(Y\setminus A) \iff X\setminus C\in v_{\mathcal{Z}}(f,Y\setminus A).
\]
Now suppose $F$ realizes $v_{\mathcal{Z}}$, that is $\dom [\deltacz,\delta_{\Sigma^0_1(Y)}]\subseteq \dom F$ and
\[
(\delta_{\Sigma^0_1(X)}\compose F)\langle\langle p,q\rangle,r\rangle\cap\delta_{\mathcal{Z}}(q) =
(\deltacz\langle p,q\rangle)^{-1}\delta_{\Sigma^0_1(Y)}(r)
\]
for all $p,r\in\Bairespc$,
$q\in\dom\delta_{\mathcal{Z}}$.  Then $F$ realizes $a_{\mathcal{Z}}$:
$\dom[\deltacz,\delta_{\Pi^0_1(Y)}]=\dom[\deltacz,\delta_{\Sigma^0_1(Y)}]\subseteq \dom F$ and any
$p,s\in\Bairespc$, $q\in\dom\delta_{\mathcal{Z}}$ have
\[
(\delta_{\Pi^0_1(X)}\compose F)\langle\langle p,q\rangle,s\rangle\cap \delta_{\mathcal{Z}}(q) =
(\deltacz\langle p,q\rangle)^{-1}(Y\setminus\delta_{\Sigma^0_1(Y)}(s))
=(\deltacz\langle p,q\rangle)^{-1}\delta_{\Pi^0_1(Y)}(s),
\]
as required.
\end{proof}


\section{Local finiteness and Dugundji systems}\label{sec:thirdagain}
Let $X$ be a topological space, suppose $\Gamma\subseteq\pwrset(X)$, and let $J$ be a nonempty set.
A collection $(A_i)_{i\in J}\in\Gamma^J$ is \emph{locally finite} (in $X$) if for every $x\in X$ there exists
$U\in\mathcal{T}_X$ such that $x\in U$ and $\{i\in J\setconstr A_i\cap U\neq\emptyset\}$ is finite.  In
applications, it happens that one wants to deal with collections which are locally finite in $Y$ for some open
$Y\subseteq X$.  Also, to formulate effective analogues of theorems about locally finite
collections, we would like to consider (as part of the names of such objects)
some witnesses of the local finiteness condition.
Specifically, in this section we will establish computable versions of the countable paracompactness property
(Lemma \ref{lem:effparac}) and the construction of a Dugundji system for a closed set $A\subseteq X$
(see Proposition \ref{prn:dugundji}); in the next section we will moreover apply the latter to
construct a retraction onto $A$ in case $X$ is effectively zero-dimensional.
These considerations motivate the following definition, where recall (from \cite{Weihrauchbook}) any numbering
$\nu:\subseteq\N\to J$ induces a numbering $\FS(\nu):\subseteq\N\to E(J)$ of the finite subsets of $J$,
namely by
\[
\FS(\nu)(k)=\{\nu(i)\setconstr i\in e(k)\}\onespace\text{ where }\onespace
\dom\FS(\nu)=\{k\setconstr e(k)\subseteq\dom\nu\}
\]
and $e=\psi^{-1}$ for the bijection $\psi:E(\N)\to\N,F\mapsto \sum_{i\in F}2^i$.
%
\begin{defi}
If $\Gamma\subseteq\pwrset(X)$ is represented by $\delta_{\Gamma}$ and
$(J,\nu)$ is a numbered set, fix a basis numbering
$\alpha:\N\to\img\alpha\subseteq\mathcal{O}(X)$ with
$(\exists\hat{a}_0)\alpha(\hat{a}_0)=\emptyset$ and define
$\LF_{\Gamma,J}$ as the set of all $(A_i)_i\in\Gamma^J$ such that
\[
(\forall x\in X)(\forall i\in J)(\exists a)(\exists S\in E(J))
\left( x\in A_i\implies x\in \alpha(a)\wedge\{j\setconstr A_j\cap\alpha(a)\neq\emptyset\}\subseteq S \right),
\]
and $\LF_{\Gamma,J}'$ as the set of all $(A_i)_i\in\LF_{\Gamma,J}$
such that $\bigcup_i A_i\in\mathcal{O}(X)$.  In all the
cases we will consider, $(J,\nu)=(\N,\id_{\N})$, and in this case it
is convenient to write $\LF_{\Gamma}$ for $\LF_{\Gamma,\N}$
($\LF_{\Gamma}'$ for $\LF_{\Gamma,\N}'$) and fix the representation
$\delta_0:=\delta_{\Gamma}^{\ordomega}|^{\LF_{\Gamma}}$ of
$\LF_{\Gamma}$; more generally we could take
e.g.~$\delta_0=[\delta_{\nu}\to\delta_{\Gamma}]|^{\LF_{\Gamma}}$.  We
also introduce the representations $\delta_1$ of $\LF_{\Gamma}$ and
$\delta_2=\delta_{\Gamma,\alpha}$ of $\LF_{\Gamma}'$ by
\begin{align*}
\langle p,q,r\rangle\in\delta_1^{-1}\{(A_j)_j\}:&\iff
p\in\delta_0^{-1}\{(A_j)_j\}\wedge \textstyle\bigcup_j A_j\subseteq
\bigcup_i\alpha(q_i)\\
&\:\:\:\:\:\:\:\:\:\:\:\wedge
(\forall i)\{j\in\N\setconstr A_j\cap\alpha(q_i)\neq\emptyset\}\subseteq \FS(\id_{\N})(r_i),\\
\langle p,q,r\rangle\in\delta_2^{-1}\{(A_j)_j\}:&\iff
\langle p,q,r\rangle\in\delta_1^{-1}\{(A_j)_j\}\wedge
\textstyle\bigcup_j A_j=\bigcup_i\alpha(q_i),
\end{align*}
\end{defi}
Thus a $\delta_1$-name for a sequence of sets $(A_i)_i\in\Gamma^{\N}$
encodes a $\delta_0$-name for the same sequence plus witness
information on the covering of each $A_i$ by (a countable collection
of) sets $\alpha(a)$ such that each intersects $A_j$ for only finitely
many $j$ (witnesses of this finiteness are also included, for each
$i$).  A $\delta_2$-name contains similar information but requires the countable cover (by sets in $\img\alpha$) to cover exactly
$\bigcup_i A_i$.

If the \emph{shift} of $\alpha$ is
$\alpha':\N\to\mathcal{O}(X)$ defined by
$\alpha'(0)=\emptyset$, $\alpha'(n+1)=\alpha(n)$ ($n\in\N$) we can also
observe $\delta_{\Gamma,\alpha}\equiv\delta_{\Gamma,\alpha'}$.  Next we
define Dugundji systems.
\begin{defi}\label{def:dugundji}
Let $(X,d)$ be a separable metric space, $\epsilon\in(0,\infty)$ and
$A\in\mathcal{A}(X)\setminus\{\emptyset\}$, where $\mathcal{A}(X)$
denotes the class of closed subsets of $X$.  A tuple $((V_i)_i,(y_i)_i)\in
\LF_{\Sigma^0_1(X),\N}\times X^{\N}$
is a \emph{Dugundji system} for $A$ (with coefficient $1+\epsilon$) provided
\begin{enumerate}
\item $\bigcup_{i\in\N}V_i=X\setminus A$, $(y_i)_i\in A^{\N}$,
\item $(\forall i)(\forall x\in V_i)(d(x,y_i)\leq (1+\epsilon) d_A(x))$, and
\item if $(n_i)_i\in\N^{\N}$ has $\lim_{i\to\infty}d(V_{n_i},A)=0$ then
$\lim_{i\to\infty}\diam V_{n_i}=0$.
\end{enumerate}
\end{defi}
Thus a Dugundji system specifies a collection of points in
$A$ and a similarly-indexed locally finite cover
$\mathcal{V}=(V_i)_{i\in I}$ of $X\setminus A$ such that the sets of
the cover are small near $A$, and each corresponding point in $A$ is
a relatively good approximation within $A$ of a given nonempty set
$V_i$.

For convenience, when the coefficient $1+\epsilon$ is not specified, we will assume $\epsilon=1$, and
follow this convention in our notation --- given a separable metric space $X$, we will denote
\[
\mathcal{D}(X):=\{((V_i)_i,(y_i)_i)\setconstr
((V_i)_i,(y_i)_i) \text{ a D.~system for }A=X\setminus\bigcup_i V_i\} \subseteq
\LF_{\Sigma^0_1(X),\N}\times X^{\N}
\]
and $\pi:\mathcal{D}(X)\to\mathcal{A}(X)\setminus\{\emptyset\},
((V_i)_i,(y_i)_i)\mapsto X\setminus\bigcup_i V_i$.  Using the standard
ball numbering $\alpha$
and $\delta_2=\delta_{\Sigma^0_1(X),\alpha'}$,
we then define
$\delta_3:\subseteq\Bairespc\to\mathcal{D}(X)$ by letting
\[
\langle p,q\rangle\in\delta_3^{-1}\{((V_i)_i,(y_i)_i)\}:\iff
p\in\delta_2^{-1}\{(V_i)_i\}\wedge q\in(\delta_X^{\ordomega})^{-1}\{(y_i)_i\},
\]
and $\delta_4=\deltad:\subseteq\Bairespc\to\mathcal{A}(X)\setminus\{\emptyset\}$ by letting
$\langle p,q\rangle\in\delta_4^{-1}\{A\}$ if there exists $D=((V_i)_i,(y_i)_i)\in\mathcal{D}(X)$ with
$\pi(D)=A$ and $\langle p,q\rangle\in\delta_3^{-1}\{D\}$.  We will mostly deal with Dugundji systems
$D\in\mathcal{D}(X)\setminus\pi^{-1}\{X\}$, i.e.~Dugundji systems for sets $A\in\mathcal{A}':=
\mathcal{A}(X)\setminus\{\emptyset,X\}$.  We will introduce further
representations related to Dugundji systems in
Section \ref{sec:newfourth} (specifically
$\delta_5$, $\delta_6$ which relate to the case $\dim A\leq 0$).

Classically, if $((V_i)_i,(y_i)_i)$ is a Dugundji system with coefficient $1+\epsilon$ for $A$, observe any
$p\in\bdry A$ has, if $p_i\in\ball{p}{2^{-i}}\setminus A$, say $p_i\in V_{n_i}$ with
$d(x,y_{n_i})\leq (1+\epsilon)d_A(x)$ for all
$x\in V_{n_i}$, that
\[
d(p,y_{n_i})\leq d(p,p_i)+d(p_i,y_{n_i})\leq d(p,p_i)+(1+\epsilon)d_A(p_i) \leq (2+\epsilon)d(p,p_i)\to 0
\]
as $i\to\infty$.  So $\bdry A\subseteq\cl{\{y_i\setconstr i\in\N\}}\subseteq A$.
Also, for any $A\in\mathcal{A}(X)\setminus\{\emptyset\}$ and dense $A'\subseteq A$, there exists
$((V_i)_i,(y_i)_i)\in\pi^{-1}\{A\}$ with $(y_i)_i\subseteq A'$.
More effectively, we will see a computable version of the construction of a Dugundji system for $A$ (with
coefficient $1+\epsilon$) in Proposition \ref{prn:dugundji},
uniformly for $A\in\mathcal{A}'$,
$\epsilon\in\Q\cap(0,1]$.  First, however, we describe some more basic results involving local finiteness.
Recall a \emph{countably paracompact} space $X$ is Hausdorff and such
that any countable open cover has a locally finite open refinement.
Other characterizations of countable paracompactness may
be found e.g.~in \cite[Thm 5.2.1]{Engelking}.
\begin{lem}[Effective countable paracompactness]\label{lem:effparac}
Suppose $(X,\mathcal{T}_X)$ is a $\Tzero$ space
with basis numbering $\alpha$,
$(\exists\tilde{a}_0)\alpha(\tilde{a}_0)=\emptyset$ and $\sqsubset$ is
a c.e.~refined inclusion of $\alpha$ with respect to $\alpha$.
Also suppose there exist computable $\hat{\alpha}:\N\to\Pi^0_1(X)$ such that
$\hat{\alpha}(\tilde{a}_0)=\emptyset$, $(\forall a)\alpha(a)\subseteq\hat{\alpha}(a)$ and
$(\forall a,b)(a\sqsubset b\implies \hat{\alpha}(a)\subseteq\alpha(b))$, and
computable $t:\N\to\Bairespc$ such that $\img t\subseteq\dom\delta_X$ and $\cl{\delta_X(\img t)}=X$,
where $\delta_X$ is the standard unpadded representation of $X$ associated to basis numbering $\alpha$.  Then the
operation $L:\Sigma^0_1(X)^{\N}\crsarr\Sigma^0_1(X)^{\N}\times\textrm{LF}_{\Pi^0_1(X)}'\times\Bairespc$
defined by
\[
L((U_i)_i)=\{((V_i)_i;(A_i)_i;e)\setconstr (\forall i)V_i\subseteq A_i\subseteq U_{e_i}\wedge
\textstyle\bigcup_i V_i=\bigcup_i U_i\}
\]
is computable.
\end{lem}
\begin{proof}
Given $p\in\delta_X^{-1}\{x\}$ and a $\delta_{\Sigma^0_1(X)}^{\ordomega}$-name
$\langle q^{(0)},\dots\rangle$ for $(U_i)_{i\in\N}$ we can dovetail searching for $a,b,i\in\N$
such that $(\exists k,l)(p_k=a\sqsubset b\sqsubset q^{(i)}_l-1)$.  Dovetail the algorithm just described over
$\delta_X$-names $p=t(n)$ where $n\in\N$, and also dovetail it with repeated output of
$\tilde{a}_0,\tilde{a}_0,0$, where $\tilde{a}_0$ is as in the statement.  In this algorithm, label $j^\Th$ outputs $a_j,b_j,i_j$ and define $e_j:=i_j$ and
\[
V_j:=\alpha(b_j)\setminus\bigcup_{i<j}\hat{\alpha}(a_i),\onespace
A_j:=\hat{\alpha}(b_j)\setminus\bigcup_{i<j}\alpha(a_i)
\]
so $V_j\subseteq A_j\subseteq U_{i_j}$ for each $j\in\N$.
If $x\in\bigcup_i\alpha(b_i)$, $i(x):=\leastmu{i}{x\in\alpha(b_i)}$, clearly
$x\in\alpha(b_{i(x)})\setminus\bigcup_{j<i(x)}\alpha(b_j)\subseteq V_{i(x)}$.
Also, for any $i,j\in\N$ with $i>j$,
\[
\alpha(a_j)\cap A_i\subseteq \alpha(a_j)\setminus\bigcup_{k<i}\alpha(a_k)=\emptyset \onespace\textrm{ holds;}
\]
thus each $x\in\alpha(a_j)$ has neighbourhood $\alpha(a_j)$ and $S=\{0,\dots,j\}$
witnessing local finiteness of collection
$(A_i)_i$ at $x$.  To restate in slightly different terms, the witness information necessary for a
$\delta_1$-name here consists of a
$\delta_{\Pi^0_1(X)}^{\ordomega}$-name
$\tilde{p}$ of $(A_j)_j$, the information $\tilde{q}:=(a_j)_{j\in\N}$ on neighbourhoods witnessing local
finiteness of $(A_j)_j$ in $\bigcup_j\alpha(a_j)$, and the information $r\in\Bairespc$ where
$r_j:=\sum_{i\leq j}2^i$
($j\in\N$) so that $(\forall j)\FS(\id_{\N})(r_j)=\{0,\dots,j\}$.  But such
$\langle\tilde{p},\tilde{q},r\rangle\in\Bairespc$ can be computed from the inputs by the above discussion.

To prove $\langle\tilde{p},\tilde{q},r\rangle$ is a $\delta_1$-name of $(A_i)_i$, it then remains to check
$\bigcup_j A_j\subseteq\bigcup_j\alpha(a_j)$ (given the above, this is also a sufficient condition for
$\langle\tilde{p},\tilde{q},r\rangle\in\delta_2^{-1}\{(A_i)_i\}$).  Clearly it is enough to show
\[
\bigcup_i U_i\subseteq\bigcup_i\alpha(a_i) \textrm{ ($\subseteq\bigcup_i\alpha(b_i)
\subseteq\bigcup_i V_i\subseteq\bigcup_i A_i\subseteq\bigcup_i U_i$).}
\]
For any $x\in U_i$ and name $p\in\delta^{-1}\{x\}$ there exist $a,b,k,l$ such that
$p_k=a\sqsubset b\sqsubset q^{(i)}_l-1$.  But then $w:=p\restrict(k+1)$ and $p':=t(n)$ for $n\in\N$ minimal
such that $t(n)\restrict(k+1)=w$ have in particular $p_k'=p_k$, so $a,b,i,k,l$ are found by the algorithm,
say as $a_j,b_j,i_j,k,l$.  In particular, $x\in\alpha(p_k')=\alpha(a_j)$, but $x$ was arbitrary.
Since then $\bigcup_i V_i=\bigcup_i U_i$ we have shown
$\langle s,\langle\tilde{p},\tilde{q},r\rangle,e\rangle\in
[\delta_{\Sigma^0_1(X)}^{\ordomega},\delta_2,\id_{\Bairespc}]^{-1}L((U_i)_i)$ for any
$s\in(\delta_{\Sigma^0_1(X)}^{\ordomega})^{-1}\{(V_i)_i\}$, and such $s$ is easy to compute from the inputs.
This completes the proof.
\end{proof}

The result of Lemma \ref{lem:effparac} can be modified to compute a locally
finite open shrinking of the original cover, in fact we have the following
for certain computable topological spaces $(X,\mathcal{T}_X,\alpha)$
(compare \cite[Remark 5.1.7]{Engelking}), where $\delta_X$ is as before:
\begin{prop}\label{prn:lfshrink}
In a $\Tzero$ space $(X,\mathcal{T}_X)$ with basis numbering $\alpha$ and
c.e.~formal inclusion $\sqsubset$ having property (\ref{inclthree}),
suppose there exists computable $\hat{\alpha}:\N\to\Pi^0_1(X)$ such that
$\alpha(a)\subseteq\hat{\alpha}(a)$ and $(a\sqsubset b\implies
\hat{\alpha}(a)\subseteq\alpha(b))$ for all $a,b\in\N$.  Also assume there
exist $\tilde{a}_0\in\N$ such that $\alpha(\tilde{a}_0)
=\hat{\alpha}(\tilde{a}_0)=\emptyset$ and computable $t:\N\to\Bairespc$
such that $\img t\subseteq\dom\delta_X$ and $\cl{\delta_X(\img t)}=X$.
Then define \(L'':\Sigma^0_1(X)^{\N}\crsarr
\Sigma^0_1(X)^{\N}\times\textrm{LF}_{\Sigma^0_1(X)}\) by letting
$L''((U_i)_i)$ take the values
\[
\bigl\{((W_i)_i;(V_i)_i)\setconstr (\exists Y\in\Sigma^0_1(X))\left(
Y=\textstyle\bigcup_i V_i=\bigcup_i U_i \wedge
(\forall i)V_i\subseteq Y\setminus W_i\subseteq U_i \right)\bigr\}.
\]
This $L''$ is $(\delta_{\Sigma^0_1(X)}^{\ordomega};
[\delta_{\Sigma^0_1(X)}^{\ordomega},\delta_{\Sigma^0_1(X),\alpha}])$-computable.
\end{prop}
\begin{proof}
First, for $(U_i)_i\in\Sigma^0_1(X)^{\N}$ and $((G_i)_i; (F_i)_i;e)\in L((U_i)_i)$ (with $L$
as in Lemma \ref{lem:effparac}) we take
$s(i):=e_i$ and $V_j:=\bigcup\{G_i\setconstr i\in\N\wedge s(i)=j\}$,
$A_j:=\bigcup\{F_i\setconstr i\in\N\wedge s(i)=j\}$.  Also let
$\langle\langle p^{(0)},\dots\rangle,q,r\rangle\in
\delta_{\Pi^0_1(X),\alpha}^{-1}\{(F_i)_i\}$.
With these definitions, any $i,j\in\N$ have
\begin{align*}
V_j\cap\alpha(q_i)\neq\emptyset & \iff (\exists t)(s(t)=j\wedge G_t\cap\alpha(q_i)\neq\emptyset)\iff
j\in s\left(\{t\setconstr G_t\cap\alpha(q_i)\neq\emptyset\}\right) \\
& \implies j\in s\left(\FS(\id_{\N})(r_i)\right),
\end{align*}
so $((V_i)_i\in\textrm{LF}_{\Sigma^0_1(X),\N}$ with witnesses given by $\langle\tilde{p}^{(0)},\dots\rangle$,
$q$ and $\tilde{r}$ where
\[
\tilde{r}_i\in\FS(\id_{\N})^{-1}\{s\left(\FS(\id_{\N})(r_i)\right)\}\onespace\text{ ($i\in\N$)}
\]
and $\tilde{p}^{(j)}$ is easily computed by dovetailing names of appropriate $G_i$ in the obvious way (dependent
on $e$).  We have that $\langle\langle \tilde{p}^{(0)},\dots\rangle,q,\tilde{r}\rangle \in
\delta_{\Sigma^0_1(X),\alpha}^{-1}\{(V_i)_i\}$ since $\bigcup_i\alpha(q_i)=\bigcup_i F_i=\bigcup_i G_i$
($=\bigcup_i U_i$) by definition of $L$.

Moreover, since each collection $\{F_i\setconstr i\in\N\wedge s(i)=j\}$ is locally finite in
$Y:=\bigcup_i\alpha(q_i)=\bigcup_i U_i$,
by \cite[Cor 1.1.12]{Engelking} we know
each $A_j$ is closed in $Y$, and we will compute $\delta_{\Sigma^0_1(X)}$-names of some $W_j$ with
$Y\setminus W_j=A_j$.  Namely, on input
$\langle\langle p^{(0)},\dots\rangle,q,r\rangle\in
\delta_{\Pi^0_1(X),\alpha}^{-1}\{(F_i)_i\}$,
$e\in\Bairespc$, $j\in\N$ we dovetail repeated output of $0$ with output of $b+1$ for each $b,k\in\N$ such that,
for $S_k:=\FS(\id_{\N})(r_k)$, we have
\[
b\sqsubset q_k\wedge(\forall i\in S_k\cap s^{-1}\{j\})(\exists l) \left(
 p^{(i)}_l\geq 1\wedge b\sqsubset p^{(i)}_l-1 \right).
\]
Any such $b$ and $k$ have
\[
\alpha(b)\subseteq \bigcap\{\alpha(q_k)\setminus F_i\setconstr i\in S_k\cap s^{-1}\{j\}\} =
\bigcap\{\alpha(q_k)\setminus F_i\setconstr i\in s^{-1}\{j\}\} = \alpha(q_k)\setminus A_j
\]
(since $\alpha(q_k)\cap F_i=\emptyset$ for all $i\in\N\setminus S_k$) while for any $x\in\alpha(q_k)\setminus A_j$ and
$i\in S_k\cap s^{-1}\{j\}$ there exists $l_i\in\N$ such that
$p^{(i)}_{l_i}\geq 1$ and $x\in\alpha(p^{(i)}_{l_i}-1)$.  Then by (\ref{inclthree}) of
Definition \ref{def:formincl} there exists $b\sqsubset q_k$ with $x\in\alpha(b)$ and
$(\forall i\in S_k\cap s^{-1}\{j\})(b\sqsubset p^{(i)}_{l_i}-1)$.

The described algorithm produces, in a multi-valued way, a $\delta_{\Sigma^0_1(X)}$-name of $Y\setminus A_j$,
which will serve as our $W_j$.  Clearly
\begin{gather*}
V_j\subseteq A_j=Y\setminus W_j\subseteq U_j \textrm{ for each $j$,}\\
\bigcup_i\alpha(q_i)=\bigcup_i G_i=\bigcup_i U_i=Y \textrm{ and }
\bigcup_i V_i\subseteq\bigcup_i U_i,
\end{gather*}
so consider $x\in\bigcup_j U_j$.  For some $i$, such $x$ has $x\in G_i\subseteq V_{s(i)}
\subseteq\bigcup_j V_j$ so indeed $\bigcup_j U_j=\bigcup_j V_j$.
This completes the description of a realizer of $L''$.
\end{proof}

\begin{cor}\label{cor:lfshrink}
Fix a computable metric space $(X,d,\nu)$, with standard ball numbering $\alpha$.  Then the operation $L':\Sigma^0_1(X)^{\N}\crsarr
\textrm{LF}_{\Sigma^0_1(X)}$ defined by
\[
L'((U_i)_i):=\{(V_i)_i\setconstr \textstyle\bigcup_i V_i=\bigcup_i U_i \wedge (\forall i)V_i\subseteq U_i\}
\]
is $(\delta_{\Sigma^0_1(X)}^{\ordomega};
\delta_{\Sigma^0_1(X),\alpha'})$-computable.
\end{cor}
\begin{proof}
Define $\hat{\alpha}':\N\to\Pi^0_1(X)$ by $\hat{\alpha}'(0):=\emptyset$ and $\hat{\alpha}'(a+1):=\hat{\alpha}(a)$
($a\in\N$).  Clearly $\hat{\alpha}'$ is $(\delta_{\N},\delta_{\Pi^0_1(X,\mathcal{T}_X,\alpha')})$-computable,
and satisfies $\alpha'(a)\subseteq\hat{\alpha}'(a)$.
We must define a c.e.~formal inclusion $\sqsubset'$ of $\alpha'$ with respect to $\alpha'$ with the properties
(\ref{inclthree}) and $(\forall a,b)(a\sqsubset' b\implies \hat{\alpha}'(a)\subseteq\alpha'(b))$, and verify
that the standard unpadded representation $\delta_{(X,\mathcal{T}_X,\alpha')}$ of $(X,\mathcal{T}_X,\alpha')$
has some computable $t:\N\to\Bairespc$ such that $\delta_{(X,\mathcal{T}_X,\alpha')}\compose t$ is total and
dense in $X$.  Well, if $\rho_{\nu}$ is the Cauchy representation,
we know
$\rho_{\nu}\leq\delta_{(X,\mathcal{T}_X,\alpha')}$, say via computable $F:\subseteq\Bairespc\to\Bairespc$
(in fact, these representations are equivalent), so $t:i\mapsto F(i^{\ordomega})$ will satisfy the requirement
on $t$.  On the other hand, $\sqsubset'$ defined by
$a\sqsubset' b:\iff a\geq 1\wedge b\geq 1\wedge a-1\sqsubset b-1$ satisfies the remaining requirements.
\end{proof}
One use of \cite[Lemma 5.1.6]{Engelking} (from whose proof Proposition \ref{prn:lfshrink} derives) is together
with Urysohn's lemma to produce partitions of unity; see \cite[Thm 5.1.9(i)$\implies$(ii)]{Engelking}.  In
zero-dimensional spaces we have an especially simple form of Urysohn's lemma, namely for any disjoint closed
$A,B\subseteq X$ there exists locally constant $g:X\to\{0,1\}\subseteq[0,1]$ with $g^{-1}\{0\}\supseteq A$ and
$g^{-1}\{1\}\supseteq B$.  Nevertheless, a nonarchimedean treatment of partitions of unity appears to be of
interest for generalizing the construction in \cite[Thm 7.6]{zerodpaper} of retractions onto nonempty closed
subsets of such spaces.  In Section \ref{sec:fourth} we will
discuss a result in this direction, using a more general Dugundji system for $A$ (and the field structure of
$K=\Omega_p$) to define the retraction in Theorem \ref{thm:b}.  This motivates the
(computable) general construction of Dugundji systems in this section (cf.~\cite[Lemma 1.2.1]{vanMill2} or
\cite[Hint to Prob 4.5.20(a)]{Engelking}).

Specifically,
we will generally consider tuples $((V_i)_i,(y_i)_i) \in \mathcal{D}(X)$;
for any computable metric space $(X,d,\nu_0)$ consider the operation
$Q:\subseteq\Q\times\mathcal{A}(X)\crsarr\mathcal{D}(X)$ defined by
\[
Q(\epsilon,A)=\{((V_i)_i,(y_i)_i)\setconstr ((V_i)_i,(y_i)_i)\in\pi^{-1}\{A\}
\text{ a Dugundji system with coefficient }1+\epsilon\},
\]
$\dom Q=\{(\epsilon,A)\setconstr \emptyset\neq A\neq X\wedge \epsilon>0\}$.

For convenience we will sometimes write
$Q(1,\cdot)=\pi^{-1}|_{\mathcal{A}'}:\subseteq\mathcal{A}(X)\crsarr\mathcal{D}(X)$.
\begin{lem}\label{lem:reorder}
For any computable metric space,
$\sigma:\subseteq\Sigma^0_1(X)^{\N}\crsarr\Bairespc$ defined by
$\dom\sigma=\{(U_i)_i\setconstr (\exists^{\infty}i)U_i\neq\emptyset\}$ and
$\sigma((U_i)_i):=\{r\setconstr \textrm{$r$ injective}\wedge \img r=\{i\setconstr U_i\neq\emptyset\}\}$
is computable.
\end{lem}
\begin{proof}[Proof sketch]
Consider the following algorithm: on input $\langle p^{(0)},\dots\rangle$, at stage $0$ let $k:=0$, at stage
$n+1$ ($n\in\N$) check whether $p^{(\pr_1 n)}_{\pr_2 n}\geq 1$; if so and $\pr_1 n$ does not appear in the output so far
(i.e.~$\pr_1 n\not\in\{r_l\setconstr l<k\}$), output $\pr_1 n$ (i.e.~let $r_k:=\pr_1 n$ and increment $k$), otherwise do nothing.
\end{proof}

\begin{lem}\label{lem:ctslindelof}
For computable metric spaces $(X,d,\nu)$, $(Z,d',\nu')$ and Cauchy
representation $\delta_Z$ of
$Z$, the computable dense sequence $z_i:=\nu'(i)$ ($i\in\N$) satisfies
\[
\bigcup_{i\in\N}u(z_i) = \bigcup_{z\in Z}u(z)
\]
for any 
continuous $u:Z\to\Sigma^0_1(X)$.
In particular,
\begin{align*}
L' &:\Cont(Z,\Sigma^0_1(X))\to\Sigma^0_1(X)^{\N},
u\mapsto (u(z_i))_{i\in\N},\\
\union &:\Cont(Z,\Sigma^0_1(X))\to\Sigma^0_1(X),
u\mapsto \bigcup_{z\in Z}u(z)
\end{align*}
are
respectively
computable.
\end{lem}
Lemma \ref{lem:ctslindelof} has been proved in \cite{zerodpaper} (note
admissibility of $\delta_Z$ and $\delta_{\Sigma^0_1(X)}$ --- see
\cite[p 62]{BrattkaPresser}, \cite{SchroderExtAdm} --- implies any
$u\in\Cont(Z,\Sigma^0_1(X))$ is
$(\delta_Z,\delta_{\Sigma^0_1(X)})$-continuous).  It
plays a similar role to the Lindel{\"o}f property of separable metric
spaces, but only for
continuous indexed covers.  The operation of continuous intersection
for closed subsets, dual to $\union$, has been considered at least in
\cite{BrattkaGherardiBorelCpxty}.
For the next result, the \emph{cylindrification} of a dense sequence
$\nu:\N\to X$ is the sequence $\lambda:\N\to X$ defined by $\lambda\langle k,l\rangle:=\nu(k)$.
\begin{prop}\label{prn:dugundji}
In a computable metric space, with $\alpha'$ as above,
$Q$ is $(\nu_{\Q},\deltar\sqcap\delta_{\textrm{dist}}^{>};\delta_3)$-computable.
\end{prop}
\begin{proof}
First consider computable metric space $(X,d,\nu_0)$, $A\in\mathcal{A}(X)\setminus\{\emptyset,X\}$ and the
cylindrification $\nu$ of $\nu_0$.  We want some choice of
$f\in\Cont_{\Sigma^0_1(X)}(X,\R_{<})$ (depending on $A$) with
\begin{equation}
X\setminus A\subseteq\dom f,\onespace 0<f(x)\leq 1\onespace
\text{ for all }
x\in X\setminus A,
\end{equation}
and, denoting
\[
b:X\to\Sigma^0_1(X),x\mapsto\begin{cases}
\ball{x}{f(x).d_A(x)}, &\mbox{\textrm{ if $x\in X\setminus A$,}}\\
\emptyset, &\mbox{\textrm{ if $x\in A$}}\end{cases}
\]
and $x_i:=\delta_X(e_i^{\ordomega})=\nu(e_i)$ ($\in X\setminus A$, $i\in\N$)
for some fixed $e\in\sigma(b\compose\nu)$ where $\sigma$ is
as in Lemma \ref{lem:reorder}
(and $b\compose\nu\in\dom\sigma$ by $A\neq X$ and
definition of $\nu$), we also would like $f$ and $t^{+}:\N\to\R$ to satisfy
\begin{equation}\label{eq:teq}
d_A(x_i)<t^{+}(i)\leq
(1+\epsilon)d_A(x_i)-(2+\epsilon)f(x_i)d_A(x_i)=:h(i)
\end{equation}
for all $i\in\N$.  Note the requirement
$d_A(x_i)<h(i)$ implies $f(x_i)<\frac{\epsilon}{2+\epsilon}$ for all
$i\in\N$ (which is stronger than the uniform bound $f(x)\leq 1$ when
$x\in\{x_i\setconstr i\in\N\}=(\img\nu)\setminus A$).

With these definitions, we first want to apply
Lemma \ref{lem:ctslindelof} with $\nu$ as
the dense sequence in the source space, $Z:=X$.  Let
$U_i:=b(x_i)=(b\compose\nu)(e_i)$ ($i\in\N$), $(W_i)_i\in L'((U_i)_i)$ for $L'$ as in
Corollary \ref{cor:lfshrink}, and take
$\langle p,q,r\rangle\in\delta_{\Sigma^0_1(X),\alpha'}^{-1}\{(W_i)_i\}$.  Then
\[
(\forall i)W_i \subseteq U_i=\ball{x_i}{f(x_i).d_A(x_i)}\subseteq X\setminus A\wedge x_i\in X\setminus A,
\]
and
\[
\bigcup_i W_i=\bigcup_i\alpha'(q_i)=\bigcup_i U_i=\bigcup_i (b\compose\nu)(i)=\bigcup_{x\in X}b(x)=X\setminus A,
\]
where we used the lower semi-continuity of $f|_{X\setminus A}$ to guarantee that $b$ is continuous
(see Remark \ref{rem:szoneadm} below).

Next, pick $y_i\in A$ such that $d(x_i,y_i)<t^{+}(i)$;
this can be done computably in $i$ if we assume information on
$t^{+}\in\R_{<}^{\N}$
is computable from the inputs (however we defer much further
discussion of computation
until the main part of the construction is given).  We claim
$((W_i)_i;(y_i)_i)\in Q(\epsilon,A)$
(indeed, $(y_i)_i\subseteq A$ regardless of properties of
$\{i\setconstr W_i\neq\emptyset\}$).

Firstly, each $i\in\N$ and any $z\in A$ have
\[
 d(z,x_i)\leq d(x_i,x)+d(x,z) \text{ for any $x\in U_i=\ball{x_i}{f(x_i).d_A(x_i)}$}
\]
hence
\begin{align*}
 \inf_{z\in A}d(z,x_i) &\leq \inf_{z\in A}\left( d(x_i,x)+d(x,z) \right)=d(x_i,x)+d_A(x)\\
\implies d_A(x_i) &\leq\inf_{x\in U_i}\left( d(x_i,x)+d_A(x) \right) \leq
\sup_{x\in U_i}d(x_i,x)+\inf_{x\in U_i}d_A(x)\\
 & \leq f(x_i).d_A(x_i)+d(A,U_i).
\end{align*}
We want also to assume $f$ is such that a
$\delta_{\Sigma^0_1(X)}^{\ordomega}$-name of $b\compose\nu$ is
computationally available from the inputs (which include a $\deltar\sqcap\delta_{\textrm{dist}}^{>}$-name of $A$);
for this we refer to Lemma \ref{lem:bnucomp} below.  For each $i$,
we observe $x\in U_i$ implies
\begin{align*}
d(x,y_i)&\leq d(x,x_i)+d(x_i,y_i)<f(x_i).d_A(x_i)+t^{+}(i) \\
&\leq f(x_i).d_A(x_i)+h(i)
\leq (1+\epsilon)(d_A(x_i)-f(x_i).d_A(x_i)) \\
&\leq (1+\epsilon)d_A(x).
\end{align*}
Finally, for any $(n_i)_i\subseteq\N$ such that $\lim_{i\to\infty}d(W_{n_i},A)=0$ and any choice of
$p_i\in U_{n_i}$ ($i\in\N$) with $\lim_{i\to\infty}d_A(p_i)=0$, note
\((\forall x\in U_i) (1-f(x_i))d_A(x_i)\leq d(A,U_i)\leq d_A(x)\)
implies
\[
\lim_{i\to\infty}d_A(x_{n_i})\leq\limsup_{i\to\infty}(1-f(x_{n_i}))^{-1}.d_A(p_i) \leq
(1-\frac{\epsilon}{2+\epsilon})^{-1}.\lim_{i\to\infty}d_A(p_i)=(1+\frac{\epsilon}{2}).0=0,
\]
so that $\diam U_{n_i}\leq 2 f(x_{n_i})d_A(x_{n_i})<2 d_A(x_{n_i})\to 0$ as $i\to\infty$.
This establishes that $((W_i)_i,(y_i)_i)$
is a Dugundji system for $A$ with coefficient $1+\epsilon$.

To summarize the above conditions sufficient to find a $\delta_3$-name of $((W_i)_i,(y_i)_i)$ computably in $A$
and $\epsilon\in\Qplus$, we will split the problem into two parts: computing (informally speaking)
$(\epsilon,A)\mapsto f$ and $(\epsilon,f)\mapsto ((W_i)_i,(y_i)_i)$.  To formalize this, consider
\[
Z_0:=\{(f,r)\in\Cont_{\Sigma^0_1(X)}(X,\R_{<})\times\Q\setconstr
(\forall i\in\nu^{-1}\dom f)(f\compose\nu)(i)<r\}
\]
and its representation $\delta_7$ defined by
\begin{align*}
a.\langle p,q\rangle\in\delta_7^{-1}\{(f,r)\} &:\iff p\in [\delta_X\to\rho_{<}]_{\Sigma^0_1(X)}^{-1}\{f\}
\wedge q\in(\deltar\sqcap\delta_{\textrm{dist}}^{>})^{-1}\{X\setminus\dom f\}\wedge\\
& a\in\nu_{\Q}^{-1}\{r\}
\wedge (\forall i\in\nu^{-1}\dom f)(f\compose\nu)(i)<r.
\end{align*}
Then consider the operation $Q':\subseteq \Q\times Z_0\crsarr\mathcal{D}(X)$ defined by
\begin{align*}
\dom Q'&=\{(\epsilon;f,\delta)\in \Q\times Z_0\setconstr \epsilon>0\wedge\emptyset\neq\dom f\neq X\wedge
\img f\subseteq(0,1]\wedge \delta=\frac{\epsilon}{2+\epsilon}\} \text{ and }\\
Q'(\epsilon;f,\delta)&:=\{D=((W_i)_i,(y_i)_i)\setconstr
D\text{ is a D.~system for $X\setminus\dom f$ with coefficient }1+\epsilon\}.
\end{align*}
This $Q'$ is $(\nu_{\Q},\delta_7; \delta_3)$-computable, namely since
we can (from the inputs) compute $t^{+}:\N\to\R_{<}$ defined by
$t^{+}(j):=\tilde{r}.d_A(x_j)$,
$\tilde{r}:=1+\epsilon-(2+\epsilon)\frac{\delta}{2} \in
\left(1,1+\epsilon-(2+\epsilon)f(x_j)\right]$ ($j\in\N$), and this
satisfies the assumption (\ref{eq:teq}) on $t^{+}$, so the
construction works uniformly in the listed input data.  Here we
compute $e\in\sigma(b\compose\nu)$ using fixed computable realizers of
$f\mapsto b\compose\nu$ and $\sigma$, and the above name(s) of $f$
(and $A:=X\setminus\dom f$).

%
On the other hand (having thus parametrized the construction of Dugundji systems by $Z_0$ according to
the operation $Q'$), it remains to
prove
\[
C:\subseteq\Q\times\mathcal{A}(X)\crsarr\Q\times Z_0,(\epsilon,A)\mapsto \{(\epsilon;f,\delta)\setconstr
\dom f=X\setminus A\wedge\img f\subseteq(0,1]\wedge
\delta=\frac{\epsilon}{2+\epsilon}\}
\]
($\dom C=\{(\epsilon,A)\setconstr \emptyset\neq A\neq X\wedge\epsilon>0\}$) is
$([\nu_{\Q},\deltar\sqcap\delta_{\textrm{dist}}^{>}];[\nu_{\Q},\delta_7])$-computable.  But this follows easily
by considering (e.g.)~the constant function $f:\subseteq X\to\R,x
\mapsto\frac{\delta}{2}$
($\dom f=X\setminus A$): recall the smn theorem guarantees, if total map
$\Q\times X\to\R,(\delta,x)\mapsto\frac{\delta}{2}$ has computable $([\delta_{\Q},\delta_X];\rho_{<})$-realizer
$H$ where $\delta_{\Q}^{-1}\{r\}=\{a.0^{\ordomega}\setconstr a\in\nu_{\Q}^{-1}\{r\}\}$, that there exists
computable total $S:\Bairespc\to\Bairespc$ such that, for all $p,q\in\Bairespc$,
\[
\left(
 q\in\dom\eta_{S(p)}\iff\langle p,q\rangle\in \dom H\right)\wedge\left(
 \langle p,q\rangle\in\dom H\implies H\langle p,q\rangle=\eta_{S(p)}(q)\right).
\]
In particular any
$q\in\dom\delta_X$ and $c\in\dom\nu_{\Q}$ have
$\langle c.0^{\ordomega},q\rangle\in\dom H$, hence
$H\langle c.0^{\ordomega},q\rangle=\eta_{S(c.0^{\ordomega})}(q)$.
Then $C$ has as an explicit realizer
$\subseteq\Bairespc\to\Bairespc,a.p\mapsto
a.c.\langle\langle S(c.0^{\ordomega}),G(p)\rangle,p\rangle$ where
we choose $c\in\nu_{\Q}^{-1}\{\frac{\nu_{\Q}(a)}{2+\nu_{\Q}(a)}\}$ computably in $a$ and $G$ is a fixed witness
of $\deltar\sqcap\delta_{\textrm{dist}}^{>}\leq\delta_{\Pi^0_1(X)}$,
cf.~\cite[Thms 3.11(1), 3.10]{BrattkaPresser}.
\end{proof}
\begin{rem}
Note the specialization of our proof to the case $\epsilon=1$ still shows
that
$((W_i)_i,(y_i)_i)\in\pi^{-1}\{A\}$.
In this case, one can take $\dom f=X\setminus A$, $f(x):=2^{-2}$, $t(i):=(d_A(x_i),\frac{5}{4}d_A(x_i))$
($i\in\N$); to verify then that $((W_i)_i,(y_i)_i)\in\pi^{-1}\{A\}$ (with witnessing information
$\langle p,q,r\rangle$), one notes
\[
\bigcup_i W_i=\bigcup_i\alpha'(q_i)=\bigcup_i U_i=\bigcup_i (b\compose\nu)(i)=\bigcup_{x\in X}b(x)=X\setminus A
\]
and completes the proof using
\[
(W_i\neq\emptyset\implies d(x_i,y_i)\leq\frac{5}{4}d_A(x_i)) \text{ and }
W_i\subseteq U_i=\ball{x_i}{2^{-2}d_A(x_i)} \text{ ($i\in\N$); }
\]
this verification can
be found in \cite{vanMill2}, and we record it here for formal completeness.
First we claim
\begin{equation}\label{eq:dugundji}
(\forall i\in\N)(\forall x\in W_i)d(x,y_i)\leq\frac{3}{2}d_A(x_i)\leq 2 d_A(x).
\end{equation}
We have
\[
d(x,y_i)\leq d(x,x_i)+d(x_i,y_i)\leq 2^{-2}d_A(x_i)+\frac{5}{4}d_A(x_i) = \frac{3}{2}d_A(x_i).
\]
Also
\[
d_A(x_i)\leq d(x_i,x)+d_A(x)\leq 2^{-2}d_A(x_i)+d_A(x)\implies\frac{3}{4}d_A(x_i)\leq d_A(x).
\]

Now if $(n_i)_i\subseteq\N$ with $\lim_{i\to\infty}d(W_{n_i},A)=0$, pick $p_i\in W_{n_i}$ such that
$\lim_{i\to\infty}d_A(p_i)=0$.  By (\ref{eq:dugundji}),
\begin{gather*}
\lim_{i\to\infty}d_A(x_{n_i}) \leq \lim_{i\to\infty}\frac{4}{3}d_A(p_i)=0, \text{ but }\\
W_{n_i}\subseteq U_{n_i}=\ball{x_{n_i}}{2^{-2}d_A(x_{n_i})}
\end{gather*}
so $\lim_{i\to\infty}\diam W_{n_i}=0$.  This completes the proof.
\end{rem}
\begin{rem}\label{rem:szoneadm}
If on the set $\mathcal{O}(X)$ of open subsets of $X$ we introduce the topology $\mathcal{T}^{\mathcal{O}}_{<}$
induced by the subbase
\[
\{\{U\in\mathcal{O}(X)\setconstr U\supseteq K\}\setconstr K\in\mathcal{K}(X)\},
\]
and assume the represented space
$\Sigma^0_1(X):=(\mathcal{O}(X),\delta_{\Sigma^0_1(X)})$ is endowed with the same topology, we can observe
$b:X\to\Sigma^0_1(X)$ in the above argument is continuous iff
$\{x\in X\setconstr b(x)\supseteq K\}$ is open for each $K\in\mathcal{K}(X)$, iff
$\{x\in X\setconstr\ball{x}{f(x)d_A(x)}\supseteq K\wedge x\in X\setminus A\}$ is open
for each $K\in\mathcal{K}^{*}(X)$, iff
$\{x\in X\setconstr x\in X\setminus A\wedge\phi_{K,A}(x)<f(x)\}$ is open for each $K\in\mathcal{K}^{*}(X)$ where
$\phi_{K,A}:\subseteq X\to\R,x\mapsto \frac{\sup_{y\in K}d(x,y)}{d_A(x)}$ ($\dom\phi_{K,A}=X\setminus A$),
and where we used that
\begin{align*}
\ball{x}{f(x)d_A(x)}\supseteq K &\iff (\forall y\in K)d(x,y)<f(x)d_A(x) \iff\phi_{K,A}(x)<f(x)
\end{align*}
provided the property $x\in X\setminus A$ is known.
But in general if $g,h:\subseteq X\to\R$ with $\dom g\cap\dom h
\supseteq B$
then $\{x\in B\setconstr g(x)<h(x)\}=\bigcup_{c\in\Q}g^{-1}(-\infty,c)\cap h^{-1}(c,\infty)\cap B$, so in the
present situation (with $\phi_{K,A}$ continuous and $f$ lower semi-continuous) we know
$\{x\in X\setconstr x\in X\setminus A\wedge\phi_{K,A}(x)<f(x)\}$ is open.  Since it is known that
$\delta_{\Sigma^0_1(X)}$ is admissible (see \cite[p 62]{BrattkaPresser}, \cite{SchroderExtAdm}), we have that
$b$ is $(\delta_X,\delta_{\Sigma^0_1(X)})$-continuous.  On the other hand, we claim
\begin{lem}\label{lem:bnucomp}
If $(X,d,\nu_0)$ is a computable metric space and $\nu$ the cylindrification of $\nu_0$,
$g:\subseteq\mathcal{A}(X)\times\Cont_{\Sigma^0_1(X)}(X,\R_{<})\to\Sigma^0_1(X)^{\N},
(A,f)\mapsto b\compose\nu$ is $(\delta_{\textrm{dist}}^{>},
[\delta_X\to\rho_{<}]_{\Sigma^0_1(X)};
\delta_{\Sigma^0_1(X)}^{\ordomega})$-computable, where
\[
b:X\to\Sigma^0_1(X),x\mapsto\begin{cases}\ball{x}{f(x)d_A(x)},\text{ if $x\in X\setminus A$,}\\
\emptyset,\text{ if $x\in A$,}\end{cases}
\]
and $\dom g=\{(A,f)\setconstr
\emptyset\neq A\neq X, \dom f\supseteq X\setminus A, f(X\setminus A)\subseteq [0,1]\}$.
\end{lem}
\begin{proof}
Consider the following algorithm: on input $\langle p,q\rangle$ and
$i\in\N$, one can dovetail output of $0^{\ordomega}$ with output of
$\langle i,j\rangle+1$ for all $j\in\N$ such that
\(\nu_{\Qplus}(j) < (f\compose\nu)(i).(d_A\compose\nu)(i)\).
Then we find the output $s^{(i)}\in\Bairespc$ has
\begin{align*}
\delta_{\Sigma^0_1(X)}(s^{(i)})&=
\bigcup\{\alpha\langle i,j\rangle\setconstr j\in\N\wedge
\nu_{\Qplus}(j)<(f\compose\nu)(i).(d_A\compose\nu)(i)\}
=(b\compose\nu)(i).
\end{align*}
But then
$G:\subseteq\Bairespc\to\Bairespc,\langle p,q\rangle\mapsto \langle s^{(0)},\dots\rangle$ is a computable
realiser as required.
\end{proof}
\end{rem}

Considering now the operation $\pi^{-1}:\mathcal{A}(X)\setminus\{\emptyset\}\crsarr\mathcal{D}(X)$
($\subseteq\LF_{\Sigma^0_1(X),\N}\times X^{\N}$), by definition we can say it is
$(\delta_4,\delta_3)$-computable (see after Definition \ref{def:dugundji}), and if we define $\delta_8$ as a
representation of $\mathcal{A}(X)\setminus\{\emptyset\}$ with
\[
\langle p,q\rangle\in\delta_8^{-1}\{A\}:\iff (\exists B\in\mathcal{A}(X)) \left(
\bdry A\subseteq B\subseteq A\wedge p\in\deltar^{-1}\{B\} \right)\wedge q\in \delta_{\Pi^0_1(X)}^{-1}\{A\},
\]
and recalling
$\mathcal{A}'=\mathcal{A}(X)\setminus\{\emptyset,X\}$, we find
\begin{prop}\label{prn:dclosedsets}
$(\deltar\sqcap\delta_{\textrm{dist}}^{>})|^{\mathcal{A}'}\leq \delta_4|^{\mathcal{A}'} \leq
\delta_8|^{\mathcal{A}'}$.
\end{prop}
\begin{proof}
The first reduction follows from computability of $Q(1,\cdot)$, while
the second is realized by \(\langle p,q\rangle\mapsto
\langle r,p\rangle\) where $r_j:=q_j+1$ ($j\in\N$).
\end{proof}

Restricted to $\mathcal{N}:=\{A\in\mathcal{A}(X)\setconstr
\interior_X A=\emptyset\}=\{A\in\mathcal{A}(X)\setconstr \bdry A=A\}$
of course $\delta_8|^{\mathcal{N}\setminus\{\emptyset\}}\equiv
(\deltar\sqcap
\delta_{\Pi^0_1(X)})|^{\mathcal{N}\setminus\{\emptyset\}}$.
Here we leave aside further description of representations of
$\mathcal{N}$ (for instance, in the special case
$X=\R$ related to the fact that, for $A\in\mathcal{A}(X)$, $A$ is nowhere dense iff
$A$ is hereditarily disconnected iff $\dim A\leq 0$); in the next section, we will be more concerned with ideas
around the use of Dugundji systems.

\section{Dugundji systems for zero-dimensional sets}\label{sec:newfourth}
In the case $(X,d,\nu_0)$ is an `effectively zero-dimensional' computable metric space it is possible to use a different construction
to give a Dugundji system $((W_i)_i,(y_i)_i)\in\pi^{-1}\{A\}$ where
$W_i$ ($i\in\N$) are pairwise disjoint open sets, at least for proper
nonempty closed subsets $A$.  Indeed, this construction can
be given for a Dugundji system relative to $A$ for $B$ whenever
$\emptyset\neq B\subseteq A\subseteq X\wedge B\neq X$ and
$\dim A\leq 0$.  This is done in Proposition \ref{prn:dugundjizd}
under certain assumptions on the (`uniformly effectively
zero-dimensional') point class $\mathcal{Y}\subseteq
\{Y\in\Pi^0_1(X)\setconstr\dim Y\leq 0\}$ with $\mathcal{Y}\ni A$,
after which Theorem \ref{thm:a} shows how to compute retractions
relative to $A$, and Proposition \ref{prn:aconverse} shows a result
in the other direction: that computing such retractions is (under some
assumptions) sufficient to show $\mathcal{Y}$ is uniformly effectively
zero-dimensional.

To be more formal, towards these aims (and since $A$ is not
necessarily effectively separable), we define
\begin{defi}
\begin{gather*}
\mathcal{D}(X,A):=\{((V_i)_i,(y_i)_i)\in\mathcal{O}(X)^{\N}\times X^{\N}
\setconstr ((V_i\cap A)_i,(y_i)_i)\in\mathcal{D}(A)\},\\
\mathcal{Y}\subseteq\mathcal{A}(X)\setminus\{\emptyset\}\text{ with }\delta_{\mathcal{Y}}\leq
\delta_{\Pi^0_1}|^{\mathcal{Y}},\onespace
E :=\bigcup_{A\in\mathcal{Y}}\{A\}\times\mathcal{D}(X,A)\subseteq
\Pi^0_1(X)\times\Sigma^0_1(X)^{\N}\times X^{\N},\\
F :=\{(A,B)\in\mathcal{Y}\times\mathcal{A}(X)\setconstr \emptyset\neq B\subseteq A\},\onespace
\pi':E\to F,(A;(V_i)_i,(y_i)_i)\mapsto (A,A\setminus\bigcup_{i\in\N}V_i),\\
\delta_5:=[\delta_{\mathcal{Y}},\deltar\sqcap\delta_{\textrm{dist}}^{>}]|^F,\onespace
\delta_6:\subseteq\N^{\N}\to E,\\
\langle p,\tilde{p},\tilde{q},r,s,t\rangle
\in\delta_6^{-1}\{(A;(V_i)_i,(y_i)_i)\}:\iff
p\in\delta_{\mathcal{Y}}^{-1}\{A\}\wedge \tilde{p}\in(\delta_{\Sigma^0_1(X)}^{\ordomega})^{-1}\{(V_i)_i\}\wedge\\
\langle \tilde{q},r,s\rangle\in\delta_{\Sigma^0_1(A),\alpha_A'}^{-1}\{(V_i\cap A)_i\}\wedge
t\in(\delta_X^{\ordomega})^{-1}\{(y_i)_i\}.
\end{gather*}
\end{defi}
Here we note $(\alpha_A)'=(\alpha')_A$ and consequently
\(\delta_{\Sigma^0_1(A)}(r)=\bigcup_{i\in\N}\alpha_A'(r_i)\).
One checks $\delta_6$ is a well-defined representation of $E$.
\begin{prop}\label{prn:dugundjizd}
If the operation $\tilde{S}:\Sigma^0_1(X)^{\N}\times\mathcal{Y}\crsarr\Sigma^0_1(X)^{\N}$ defined by
\[
\tilde{S}((V_i)_i,Y)=\{(W_i)_i\setconstr(\forall i)W_i\subseteq V_i\wedge\textstyle
\bigcup_i W_i\cap Y=\bigcup_i V_i\cap Y
\wedge (\forall i,j)(i\neq j\implies W_i\cap W_j=\emptyset)\}
\]
is 
computable then the operation
$R':\subseteq F\crsarr E$ is $(\delta_5;\delta_6)$-computable, where

\(\dom R'=\{(A,B)\in F\setconstr B\neq X\}
=\{(A,B)\in\mathcal{Y}\times\mathcal{A}(X)\setconstr
\emptyset\neq B\subseteq A\wedge B\neq X\}\) and
\[
R'(A,B):=\{(A;(W_i)_i,(y_i)_i)\setconstr
\pi'(A;(W_i)_i,(y_i)_i)=(A,B)\text{ and $(W_i)_i$ pairwise disjoint}\}.
\]
\end{prop}
\begin{rem}\label{rem:unifzd}
The computability of $\tilde{S}$ as above (for a particular represented class
$\mathcal{Y}\subseteq\{A\subseteq X\setconstr \dim A\leq 0\}$)
can be shown to follow from computability of $B$ as in \cite[\S 4]{zerodpaper} almost identically to the
proof given there for ($B$ computable $\implies$ $S$ computable).  In fact, one has
($B$ computable $\implies$ $\tilde{S}$ computable $\implies$ $S$ computable) and, in the case that
$\mathcal{Y}\subseteq\{A\in\mathcal{A}(X)\setconstr \dim A\leq 0\}$ with
$\delta_{\mathcal{Y}}\leq\delta_{\Pi^0_1(X)}|^{\mathcal{Y}}$, one can show these conditions as well as the
separate conditions of computability of $N,M$ are all equivalent (see Section \ref{sec:zdsubsets} below).
Here we suggest to call such represented spaces
$(\mathcal{Y},\delta_{\mathcal{Y}})$, where $S$ is computable, \emph{uniformly effectively zero-dimensional}.
Two examples (to be shown elsewhere) are $\mathcal{Y}=\{K\in\mathcal{K}(X)\setconstr \dim K\leq 0\}$ with
$\delta_{\mathcal{Y}}=\deltac|^{\mathcal{Y}}$, and
$(\mathcal{Y},\delta_{\mathcal{Y}})=\Pi^0_1(\N^{\N})$.
\end{rem}
\begin{proof}[Proof of Proposition \ref{prn:dugundjizd}]
First, observe from Proposition \ref{prn:dugundji} that $\pi^{-1}|_{\mathcal{A}'}$ is computable.  After this,
we consider the operations
\[
f=(\id_{\mathcal{Y}}\times\pi^{-1})|_F:F\crsarr\mathcal{Y}\times\mathcal{D}(X),
\onespace I:\subseteq\mathcal{Y}\times\mathcal{D}(X)\to E,\onespace
R:E\crsarr E
\]
defined by
\begin{gather*}
f(A,B):=\{A\}\times\pi^{-1}(B),\onespace I(A;(V_i)_i,(y_i)_i):=(A;(V_i)_i,(y_i)_i)\in E,\\
\text{ and }\onespace
R(A;(V_i)_i,(y_i)_i):=\{(A;(W_i)_i,(z_i)_i)\in E\setconstr (W_i)_i\in \tilde{S}((V_i)_i,A)\},
\end{gather*}
where $\dom I:=\{(A,D)\setconstr \pi(D)\subseteq A\}$.

We find $f$ is well-defined (by existence of a Dugundji system for arbitrary
$B\in\mathcal{A}(X)\setminus\{\emptyset\}$), while $I$ is well-defined by checking any Dugundji system
$((V_i)_i,(y_i)_i)$ for $B$ and any $\mathcal{Y}\ni A\supseteq B$
have $((V_i\cap A)_i,(y_i)_i)$ a Dugundji system for $B$ in $A$.  Namely, if
$p\in\delta_{\mathcal{Y}}^{-1}\{A\}$ and $q=\langle\langle \tilde{q},r,s\rangle,t\rangle\in
\delta_3^{-1}\{((V_i)_i,(y_i)_i)\}$, we claim $(A;(V_i)_i,(y_i)_i)\in E$ and
$\langle p,\tilde{q},\tilde{q},r,s,t\rangle\in\delta_6^{-1}\{(A;(V_i)_i,(y_i)_i)\}$; to see this, note
$\tilde{q}\in(\delta_{\Sigma^0_1(X)}^{\ordomega})^{-1}\{(V_i)_i\}$,
$\langle\tilde{q},r,s\rangle\in\delta_{\Sigma^0_1(A),\alpha_A'}^{-1}\{(V_i\cap A)_i\}$ and
$t\in(\delta_X^{\ordomega})^{-1}\{(y_i)_i\}$ with the second claim using
$\tilde{q}\in(\delta_{\Sigma^0_1(A)}^{\ordomega})^{-1}\{(V_i\cap A)_i\}$,
$\bigcup_i(V_i\cap A)=\delta_{\Sigma^0_1(A)}(r)$ and
\[
\{j\setconstr V_j\cap A\cap\alpha_A'(r_i)\neq\emptyset\}\subseteq
\{j\setconstr V_j\cap\alpha'(r_i)\neq\emptyset\}\subseteq\FS(\id_{\N})(s_i)
\]
for all $i\in\N$.

Clearly here $\langle p,q\rangle\mapsto \langle p,\tilde{q},\tilde{q},r,s,t\rangle$ computably realises $I$.
On the other hand, $R$ is well-defined since if $((V_i)_i,(y_i)_i)$ is a Dugundji system relative to
$A\in\mathcal{Y}$ and $(W_i)_i\in\tilde{S}((V_i)_i,A)$, then (we claim) $((W_i)_i,(y_i)_i)$ is a Dugundji
system relative to $A$:\\
if $\langle p,\tilde{p},\tilde{q},r,s,t\rangle\in
\delta_6^{-1}\{(A;(V_i)_i,(y_i)_i)\}$
and $G$ is a computable realizer of $\tilde{S}$, we get
$\langle \tilde{q},r,s\rangle\in
\delta_{\Sigma^0_1(A),\alpha_A'}^{-1}\{(V_i\cap A)_i\}$ and $\tilde{p}':=G\langle\tilde{p},p\rangle\in
(\delta_{\Sigma^0_1(X)}^{\ordomega})^{-1}\{(W_i)_i\}$ hence $\tilde{p}'\in
(\delta_{\Sigma^0_1(A)}^{\ordomega})^{-1}\{(W_i\cap A)_i\}$, while $W_i\subseteq V_i$ implies
\[
\{j\setconstr W_j\cap A\cap\alpha_A'(r_i)\neq\emptyset\}\subseteq
\{j\setconstr V_j\cap A\cap\alpha_A'(r_i)\neq\emptyset\}\subseteq\FS(\id_{\N})(s_i)
\]
for each $i$, so
$(W_i\cap A)_i\in\LF_{\Sigma^0_1(A)}$ with
$\langle\tilde{p}',r,s\rangle\in
\delta_{\Sigma^0_1(A),\alpha_A'}^{-1}\{(W_i\cap A)_i\}$ (we
also used that \(\bigcup_i V_i\cap A=\delta_{\Sigma^0_1(A)}(r)\implies \bigcup_i W_i\cap A=\delta_{\Sigma^0_1(A)}(r)\)).  Finally,
\begin{gather*}
(y_i)_i\subseteq B:=A\setminus\bigcup_i(V_i\cap A)=A\setminus\bigcup_i(W_i\cap A),
\onespace
(\forall i)(\forall x\in W_i\cap A)(d(x,y_i)\leq 2 d_B(x)),
\end{gather*}
and for each $(n_i)_i\in\Bairespc$,
\begin{align*}
\lim_{i\to\infty}d(W_{n_i}\cap A,B)=0 & \implies
\lim_{i\to\infty}d(V_{n_i}\cap A,B)=0\\
& \implies \lim_{i\to\infty}\diam(W_{n_i}\cap A)\leq
\lim_{i\to\infty}\diam(V_{n_i}\cap A)=0,
\end{align*}
using that $(\forall i)(W_i\cap A\subseteq V_i\cap A)$.
Thus $((W_i\cap A)_i,(y_i)_i)\in\mathcal{D}(A)$ as required,
and \(\langle p,\tilde{p}',\tilde{p}',r,s,t\rangle\in
\delta_6^{-1}\{(A;(W_i)_i,(y_i)_i)\}\).

To summarize the above facts:
\begin{enumerate}
\item $f|_{\{(A,B)\setconstr \emptyset\neq B\neq X\}}$ is
$([\delta_{\mathcal{Y}},\deltar\sqcap\delta_{\textrm{dist}}^{>}],[\delta_{\mathcal{Y}},\delta_3])$-computable,
\item $I$ is $([\delta_{\mathcal{Y}},\delta_3],\delta_6)$-computable, and
\item since $\tilde{S}:\Sigma^0_1(X)^{\N}\times\mathcal{Y}\crsarr\Sigma^0_1(X)^{\N}$ is
$([\delta_{\Sigma^0_1(X)}^{\ordomega},\delta_{\mathcal{Y}}],\delta_{\Sigma^0_1(X)}^{\ordomega})$-computable,
$R$ is $(\delta_6,\delta_6)$-computable.
\end{enumerate}
Since $(R\compose I\compose f)(A,B)\subseteq R'(A,B)$ for all $(A,B)\in\dom R'$, we should have computability
of $R'$.  More formally, if $(A,B)\in\dom R'$ then $\emptyset\neq B\neq X$ and $B\subseteq A$, with any
$D=((V_i)_i,(y_i)_i)\in\pi^{-1}\{B\}$ having $(A,D)\in\dom I$ (hence $f(A,B)\subseteq\dom I$).  On the other
hand, if $(A',D')\in R(A,D)$ with $D'=((W_i)_i,(z_i)_i)$ ($\in\mathcal{D}(X,A')$) then $A'=A$,
$\bigcup_i W_i\cap A=\bigcup_i V_i\cap A$ and $(W_i)_i$ pairwise disjoint.  But then
$A\setminus\bigcup_i W_i=A\setminus\bigcup_i V_i=B$, so $(A',D')\in R'(A,B)$.  This completes the proof.
\end{proof}
%
%
%

Using this construction of a Dugundji system for $B$ relative to $A$ with pairwise disjoint sets, we now present
a version of \cite[Thm 7.6]{zerodpaper} relative to $A\in\mathcal{Y}$, showing it is possible to construct retractions onto specified closed subsets $\emptyset\neq B\subseteq A$ with $B\neq X$.  The statement of the
theorem when $A=X$ is somewhat stronger than in \cite{zerodpaper}, as it does not require compactness, and the
proof is simplified
in some respects.  Further simplifications of the construction of a Dugundji system with pairwise disjoint sets
in the case of an effectively zero-dimensional computable metric space $X$ are also possible, and are discussed
after the theorem; we just show, in two directions, how to obtain clopen sets here, and how to avoid the
restriction $B\neq X$ in certain cases.
\begin{thm}\label{thm:a}
Suppose $X$ is a computable metric space with $\mathcal{Y}\subseteq\{A\in\mathcal{A}(X)\setconstr \dim A\leq 0\}$
a represented class such that $\delta_{\mathcal{Y}}\leq\delta_{\Pi^0_1(X)}|^{\mathcal{Y}}$ and $\tilde{S}$
is computable.  Then
\[
E':\subseteq F\crsarr \Cont_{\Pi^0_1(X)}(X,X),(A,B)\mapsto
\{f\setconstr \img f=B\wedge f|_B=\id_B\wedge\dom f=A\}
\]
($\dom E'=\{(A,B)\in F\setconstr B\neq X\}
=\{(A,B)\in\mathcal{Y}\times\mathcal{A}(X)\setconstr \emptyset\neq B\subseteq A\wedge B\neq X\}$) is
$(\delta_5;[\delta_X\to\delta_X]_{\Pi^0_1(X)})$-computable.
\end{thm}
\begin{proof}
First take $(A;(W_i)_i,(y_i)_i)\in R'(A,B)$.
We define
\[
f:\subseteq X\to X,x\mapsto \begin{cases}
x, &\mbox{\textrm{ if $x\in B$}}\\
y_i, &\mbox{\textrm{ if $(\exists i)W_i\cap A\ni x$}}\end{cases}
\]
($\dom f=A$), and will show $f^{-1}V\in\Sigma^0_1(A)$ (for an
arbitrary $V\in\Sigma^0_1(X)$) computably uniformly in $\tilde{q},r,s$
(where $q^{(i)}\in\delta_{\Sigma^0_1(X)}^{-1}\{W_i\}$,
$t^{(i)}\in\delta_X^{-1}\{y_i\}$,
$\tilde{q}=\langle\langle q^{(0)},\dots\rangle,
\langle t^{(0)},\dots\rangle\rangle$,
$r\in\delta_{\Sigma^0_1}^{-1}\{V\}$ and $s\in\deltar^{-1}\{B\}$), in the sense we can compute a
$\delta_{\Sigma^0_1(X)}$-name of some $U$ with $U\cap A=f^{-1}V$.

Indeed, we will define computable
$F:\subseteq\Bairespc\to\Bairespc,\langle\langle\langle p,\tilde{q}\rangle,r\rangle,s\rangle\mapsto h$ such that each
relevant choice of $\tilde{q},r,s$ gives a function
$u=\delta_{\Sigma^0_1(X)}\compose F\langle\langle\langle\cdot,\tilde{q}\rangle,r\rangle,s\rangle:\dom\delta_X\to
\Sigma^0_1(X)$ satisfying
\begin{align*}
\delta_X(p) \in u(p)\cap A \subseteq f^{-1}\delta_{\Sigma^0_1(X)}(r),\onespace &\text{ if }
p\in\delta_X^{-1}f^{-1}\delta_{\Sigma^0_1(X)}(r);\\
\onespace
u(p)\cap A =\emptyset,\onespace &\text{ if }p\in\dom\delta_X\setminus\delta_X^{-1}f^{-1}\delta_{\Sigma^0_1(X)}(r).
\end{align*}
Lemma \ref{lem:ctslindelof} will give $f^{-1}\delta_{\Sigma^0_1}(r)=\bigcup_{p\in\dom\delta_X}u(p)\cap A$
and a $\delta_{\Sigma^0_1}$-name of
$\bigcup_{p\in\dom\delta_X}u(p)=\bigcup_i(u\compose\nu)(i)$ is available
from the inputs.

To define $h$, we dovetail repeated output of `$0$' with searching for $j,M,n$ such that
$a:=\langle p_j,\overline{2^{-j+1}}\rangle$ satisfies $R_0(p,\tilde{q},r,j)\vee R_1(p,r,s,j,M,n)$, outputting `$a+1$'
followed by $0^{\ordomega}$ if such are found.  Here, to be specific, we test the conditions
\begin{align*}
R_0(p,\tilde{q},r,j) &:\equiv (\exists i,k,l,m)\left( q^{(i)}_k\geq 1\wedge a\sqsubset q^{(i)}_k-1\wedge
r_l\geq 1\wedge \langle t^{(i)}_m,\overline{2^{-m+1}}\rangle\sqsubset r_l-1 \right),\\
R_1(p,r,s,j,M,n) &:\equiv
(\exists k)\left( r_k\geq 1\wedge \left\langle p_j,\overline{2^{-j+1}
+\frac{2}{M+1}}\right\rangle\sqsubset r_k-1\right) \\
&\wedge d(\nu(p_j),z_n)+2^{-j+1}<(M+1)^{-1}
\end{align*}
where $(z_i)_{i\in\N}\subseteq B$ is defined (referring to definition of $\deltar$ and recalling
$B\neq\emptyset$) by $z_i:=\delta_X(P s^{(i)})$ ($i\in\N$) for $\langle s^{(0)},\dots\rangle=s$.

Any nonzero output $a+1$ must have the form $a=\langle p_j,\overline{2^{-j+1}}\rangle$ for some $j$,
with either (if $R_0(p,\tilde{q},r,j)$ holds)
$f(\alpha(a)\cap A)\subseteq f(W_i\cap A)\subseteq\{y_i\}\subseteq V$ for appropriate $i$,
or else (if $R_1(p,r,s,j,M,n)$ holds) we can argue as follows.  Any $z\in\alpha(a)\cap A$ has
either $z\in B$ or $(\exists i)W_i\ni z$.  In the first case,
\[
f z=z\in\alpha(a)\cap A\subseteq
\alpha\left\langle p_j,\overline{2^{-j+1}+\frac{2}{M+1}}\right\rangle\subseteq V
\]
(by the first clause of $R_1(p,r,s,j,M,n)$).  In the second case, if
$\nbhd{B}{\epsilon}:=d_B^{-1}(-\infty,\epsilon)$ then
\[
z\in\alpha(a)\cap A\subseteq\ball{z_n}{(M+1)^{-1}}\subseteq
\nbhd{B}{(M+1)^{-1}}
\]
(by second clause of $R_1(p,r,s,j,M,n)$), so $d_B(z)<(M+1)^{-1}$ and
\begin{align*}
d(f z,\nu(p_j)) &=d(y_i,\nu(p_j))\leq d(y_i,z)+d(z,\nu(p_j))\leq 2 d_B(z)+2^{-j+1}<\frac{2}{M+1}+2^{-j+1}\\
&\implies f z\in \alpha\left\langle p_j,\overline{\frac{2}{M+1}+2^{-j+1}}\right\rangle\subseteq V
\end{align*}
where we used the definition of a Dugundji system and the first clause of $R_1(p,r,s,j,M,n)$.  Thus
($\delta_{\Sigma^0_1(X)}(t)=\alpha(a)$ and) $\alpha(a)\cap A\subseteq f^{-1}V$ for any nonzero output $a+1$.
Contrapositively, if $x:=\delta_X(p)\in X\setminus f^{-1}V$ then the output
must be $h=0^{\ordomega}$ to avoid a contradiction (since
$x\in\alpha\langle p_j,\overline{2^{-j+1}}\rangle$ for all $j\in\N$).

On the other hand any $x\in f^{-1}V$ either has $(\exists i)W_i\ni x$ or else $x\in B$.
In the former case $y_i=f x\in V$ and we know there exist $k,l$ such that
\[
q^{(i)}_k\geq 1\wedge r_l\geq 1
\wedge x\in\alpha(q^{(i)}_k-1)\wedge y_i\in\alpha(r_l-1).
\]
Using continuity of $d$ with $\nu(t^{(i)}_m)\to y_i$ and $2^{-m+1}\to 0$ as $m\to\infty$, we get
\[
d(\nu(\pr_1(r_l-1)),\nu(t^{(i)}_m))+2^{-m+1}<\nu_{\Qplus}(\pr_2(r_l-1))
\]
for large $m$, similarly
$a:=\langle p_j,\overline{2^{-j+1}}\rangle\sqsubset q^{(i)}_k-1$ for large $j$.  So $R_0(p,\tilde{q},r,j)$ holds
under our assumption, and we may assume instead $x\in B$.  Then $V\ni f x=x$ and we can pick $k,M$ such that
\[
r_k\geq 1 \text{ and }d(x,\nu(\pr_1(r_k-1)))+\frac{2}{M+1}<\nu_{\Qplus}(\pr_2(r_k-1)),
\]
then pick $j\in\N$ such
that $\langle p_j,\overline{2^{-j+1}+\frac{2}{M+1}}\rangle\sqsubset r_k-1$ and
$d(\nu(p_j),x)+2^{-j+1}<(M+1)^{-1}$, then finally pick $n\in\N$ such that $z_n$ is sufficiently close to $x$
that $d(\nu(p_j),z_n)+2^{-j+1}<(M+1)^{-1}$.  One checks $R_1(p,r,s,j,M,n)$ holds in this case.  As a result,
the algorithm outputs $a+1$ for some $a$ with $A\cap\alpha(a)\ni x$, provided $x\in f^{-1}V$, and
so $\delta_{\Sigma^0_1(X)}(h)\cap A=f^{-1}V$ obtains.
The computable function $F:\subseteq\Bairespc\to\Bairespc$ thus has the properties claimed.
\end{proof}

\begin{rem}\label{rem:zddgone}
If $X$ is an effectively zero-dimensional computable metric space,
$\epsilon=1$ and
$A\in\mathcal{A}(X)\setminus\{\emptyset,X\}$, let $u:X\to\Sigma^0_1(X)$ be defined as before,
and let $e\in\sigma(u\compose\nu)$ where $\nu$ is the cylindrification of $\nu_0$.
Next let $x_i:=\nu(e_i)$, $U_i:=u(x_i)=\ball{x_i}{2^{-2}d_A(x_i)}$ ($i\in\N$)
and $(W_i)_i\in\tilde{S}^X((U_i)_{i\in\N})$, where
$\tilde{S}^X:\Sigma^0_1(X)^{\N}\crsarr\Sigma^0_1(X)^{\N}$ is
(as in \cite[\S 5]{zerodpaper}) defined by
\[
\tilde{S}^X((U_i)_i) = \{(W_i)_i\setconstr (W_i)_i \text{ pairwise disjoint with }W_i\subseteq U_i\text{ and }
\textstyle\bigcup_i W_i=\bigcup_i U_i\}.
\]
With $V$ defined by $V:\Sigma^0_1(X)\crsarr\Delta^0_1(X)^{\N},U\mapsto
\{(W_i^{*})_i\setconstr U=\dot{\bigcup}_i W_i^{*}\}$
we want to take $(\check{W}_{\langle i,j\rangle})_j\in V(W_i)$ for each $i$ and pick
$\check{y}_{\langle i,j\rangle}\in A$ such that
$d(x_i,\check{y}_{\langle i,j\rangle})<\frac{5}{4}d_A(x_i)$ for all $i,j\in\N$ (this is again computable in the
inputs since a $\deltar\sqcap\delta_{\textrm{dist}}^{>}$-name of $A$ is assumed available).  The verification
that $((\check{W}_i)_i;(\check{y}_i)_i)\in \pi^{-1}(A)$ and that (informally speaking)
$A\mapsto ((\check{W}_i)_i;(\check{y}_i)_i)$ is $(\deltar\sqcap\delta_{\textrm{dist}}^{>};
[\delta_{\Delta^0_1(X),\alpha'},\delta_X^{\ordomega}])$-computable will follow from the
mentioned inequality and $W_i\subseteq U_i=\ball{x_i}{2^{-2}d_A(x_i)}$ as before, once we compute some
$p=\langle p^{(0)},\dots\rangle\in(\delta_{\Delta^0_1(X)}^{\ordomega})^{-1}\{(\check{W}_i)_i\}$ and
$q,r\in\Bairespc$ appropriately such that $\langle p,q,r\rangle \in
\delta_{\Delta^0_1(X),\alpha'}^{-1}\{(\check{W}_i)_i\}$.

%
So, if basis numbering $\alpha'$ is defined from $\alpha$ as before, and
$h\in\trfns$ is such that $\img h=\{\langle m,n\rangle\setconstr m,n\in\N\wedge m\sqsubset n\}$,
on input $\langle p^{(0)},\dots\rangle$ we enumerate in $q\in\Bairespc$ all $a+1$ ($a\in\N$) such that
$(\exists k,l)(p^{(k)}_{2 l}\geq 1\wedge a\sqsubset p^{(k)}_{2 l}-1)$, along with (in dovetail fashion)
$\ordomega$ copies of $0$.
Any $i,j$ such that
$\alpha'(q_i)\cap\check{W}_j\neq\emptyset$ have, for corresponding $k$ such that
$\alpha'(q_i)\subseteq\check{W}_k$, that necessarily $\check{W}_k\cap\check{W}_j\neq\emptyset\implies k=j$, so
$\{j\setconstr \check{W}_j\cap\alpha'(q_i)\neq\emptyset\}\subseteq \FS(\id_{\N})(r_i)$ where by definition
$r_i\in\FS(\id_{\N})^{-1}\{k'\}$ for
\[
k'=\pi^{(3)}_1\leastmu{\langle k,l,m\rangle}{p^{(k)}_{2 l}\geq 1\wedge h(m)=\langle q_i,p^{(k)}_{2 l}-1\rangle}.
\]
Using the definition of $\FS(\id_{\N})$, some such $r_i$ is uniformly computable from the inputs and $i$.

All that
remains to compute $((\check{W}_i)_i;(\check{y}_i)_i)$ from the inputs is to check computability of $V$.
\begin{proof}[Proof that $V$ is computable]
Since $X$ is effectively zero-dimensional, there exist computable
$b:\N\to\Delta^0_1(X)$ and c.e.~formal inclusion $\sqsubset'$ of $b$
with respect to $\alpha$ (the latter is the standard ball numbering of
$(X,d,\nu_0)$), such that $b$ is a basis numbering.  Then, similarly
to \cite[Proof of Prop 7.3(2)$\implies$(3)]{zerodpaper}, if
$F:\subseteq\Bairespc\to\Bairespc$ is a computable realizer of $b$ and $h\in\trfns$
such that $\img h=\{\langle a,c\rangle\setconstr a\sqsubset' c\}$, we define
$G:\Bairespc\to\Bairespc,p\mapsto\langle p^{(0)},\dots\rangle$ where
\[
p^{(i)}:=\begin{cases}
F(m.0^{\ordomega}), &\mbox{\textrm{ if $i=\langle l,m,n\rangle^{(3)}\wedge p_n\geq 1\wedge
h(l)=\langle m,p_n-1\rangle$,}}\\
\langle 0^{\ordomega},\id_{\N}\rangle, &\mbox{\textrm{ otherwise}}.
\end{cases}
\]
This $G$ is computable and realizes the operation
$\Sigma^0_1(X)\crsarr\Delta^0_1(X)^{\N},U\mapsto\{(W_i)_i\setconstr U=\bigcup_i W_i\}$, after which we take
$W_i^{*}:=W_i\setminus\bigcup_{j<i}W_j$ to get pairwise disjointness.
\end{proof}
\end{rem}
\begin{rem}\label{rem:zddgtwo}
To avoid the condition $A\neq X$ in the above argument, we may suppose in place of $V$ that
$V':\Sigma^0_1(X)^{\N}\crsarr\Sigma^0_1(X)^{\N}\times\{0,1\}^{\N}$, defined by letting
$V'((U_i)_i)$ be the set
\[
\{((\check{W}_k)_k,s)\setconstr (\check{W}_k)_k \text{ p.w.~disj, }
(\forall i,j)\check{W}_{\langle i,j\rangle}\subseteq U_i,
(\forall k)(\check{W}_k=\emptyset\smash{\iff} s_k=0),\textstyle\bigcup_k\check{W}_k=\bigcup_i U_i\},
\]
is computable.  If this holds for a particular effectively zero-dimensional computable metric space
$(X,d,\nu_0)$, we can take $x_i:=\nu_0(i)$, $U_i:=u(x_i)=\ball{x_i}{2^{-2}d_A(x_i)}$ (possibly with zero
radius), $((\check{W}_i)_i,s)\in V'((U_i)_{i\in\N})$ (so $\check{W}_{\langle i,j\rangle}\subseteq U_i$)
and pick
\[
\begin{cases}
\check{y}_{\langle i,j\rangle}\in A\text{ arbitrarily}, &\mbox{\textrm{ if $s_{\langle i,j\rangle}=0$,}}\\
\check{y}_{\langle i,j\rangle}\in A\mathst d(\check{y}_{\langle i,j\rangle},x_i)<\frac{5}{4}d_A(x_i),
&\mbox{\textrm{ if $s_{\langle i,j\rangle}=1$}}
\end{cases}
\]
(here note $s_{\langle i,j\rangle}=1\iff\check{W}_{\langle i,j\rangle}\neq\emptyset\implies u(x_i)\neq
\emptyset\iff x_i\not\in A$).

Then one observes $(\check{W}_{\langle i,j\rangle}\neq\emptyset\implies d(x_i,\check{y}_{\langle i,j\rangle})
\leq\frac{5}{4}d_A(x_i))$ for all $i,j$, so the verification of the properties of
$((\check{W}_i)_i,(\check{y}_i)_i)$ follows as before (and it is a Dugundji system).  To extend the
domain of $\pi^{-1}|_{\mathcal{A}'}$ to include $A=X$ in this situation is easy;
indeed note to compute a
$\delta_3$-name
$\langle\langle p,q,r\rangle,t\rangle$ of $((\check{W}_i)_i;(\check{y}_i)_i)\in \pi^{-1}(A)$ we only require
$\bigcup_i\alpha'(q_i)=\bigcup_i\check{W}_i=X\setminus A$,
$\{j\setconstr \check{W}_j\cap\alpha'(q_i)\neq\emptyset\}\subseteq\FS(\id_{\N})(r_i)$ for all $i\in\N$
and $t\in(\delta_X^{\ordomega})^{-1}\{(\check{y}_i)_i\}$.
So, compute $q\in\Bairespc$ (similarly to before) as follows: on input
$\langle p^{(0)},\dots\rangle\in (\delta_{\Sigma^0_1(X)}^{\ordomega})^{-1}\{(\check{W}_i)_i\}$,
enumerate $a+1$ for all $a\in\N$ such that $(\exists k,l)(p^{(k)}_l\geq 1\wedge a\sqsubset p^{(k)}_l-1)$,
dovetailed with output of $\ordomega$ copies of $0$.  We find of course that any
$x\in X\setminus A=\bigcup_k\check{W}_k$ has some $k,l$ such that $x\in\alpha'(p^{(k)}_l)$, and then some $j$
such that $p^{(k)}_l\geq 1\wedge q_j\geq 1\wedge q_j-1\sqsubset p^{(k)}_l-1\wedge x\in\alpha'(q_j)$.
Also compute $r\in\Bairespc$ (similarly to before) by
$r_i\in\FS(\id_{\N})^{-1}(\{k'\})$ for
\[k'=
\pi^{(3)}_1\leastmu{\langle k,l,m\rangle}{p^{(k)}_l\geq 1\wedge h(m)=\langle q_i,p^{(k)}_l-1\rangle}.
\]
\end{rem}
\begin{rem}
Returning to the assumption of computability of $V'$, it is not too hard to see
(modifying the proof of \cite[Prop 4.1(iii)$\implies$(iv)]{zerodpaper}) conditions under which this will
follow.  Namely, if the effectively zero-dimensional computable metric space $X$ has computable basis
numbering $b:\N\to\Delta^0_1(X)$ and c.e.~refined inclusion $\sqsubset'$ of $b$ with respect to
$\alpha$, we might suppose
$X$ has the \emph{effective covering property} of $b$ with respect to
$b$ (cf.~\cite[Defn 2.6(1)]{BrattkaPresser}) and that
$A_b:=\{a\in\N\setconstr b(a)=\emptyset\}$ is c.e., i.e.~suppose
\[
H:=\{\langle a,\langle w\rangle\rangle\setconstr a\in\N, w\in\N^{*} \text{ and }
b(a)\subseteq\bigcup_{i<\absval{w}}b(w_i)\}\onespace\text{ is c.e.}
\]
This property holds for example for the Baire space with the
usual metric and $b=\alpha$, as well as for any computably compact effectively zero-dimensional computable
metric space, and it implies computability of $V'$ as mentioned.
%
Namely, on input $\langle p^{(0)},\dots\rangle\in (\delta_{\Sigma^0_1(X)}^{\ordomega})^{-1}\{(U_i)_i\}$ we
compute some $q^{(i)}\in\Bairespc$ such that $U_i=\bigcup_j b(q^{(i)}_j)$, then compute
$W^{(i)}_k:=b(q^{(i)}_k)\setminus\bigcup_{j<k}b(q^{(i)}_j)\in\Delta^0_1(X)$ ($i,k,\in\N$).  We have
$\dot{\bigcup}_k W^{(i)}_k=U_i$ and the property $W^{(i)}_k=\emptyset$ is decidable in $i,k$, using
$W^{(i)}_k=\emptyset\iff \langle q^{(i)}_k,\langle q^{(i)}\restrict k\rangle\rangle\in H$ to semi-decide
emptiness.
\end{rem}

We will now give a result in the converse direction to
Theorem \ref{thm:a}, thus establishing that computability of $E'$ (in a
computable metric space with more than one point) is an equivalent
condition for the uniform effective zero-dimensionality of
$\mathcal{Y}$; here recall:
\[
E':\subseteq F\crsarr \Cont_{\Pi^0_1(X)}(X,X),(A,B)\mapsto
\{f\setconstr \img f=B\wedge f|_B=\id_B\wedge\dom f=A\}
\]
($\dom E'=\{(A,B)\in F\setconstr B\neq X\}
=\{(A,B)\in\mathcal{Y}\times\mathcal{A}(X)\setconstr \emptyset\neq B
\subseteq A\wedge B\neq X\}$).
We note the proof relies on the material in \cite[\S 8]{zerodpaper} on
bilocated subsets; we refer the reader there (or to the treatment of
this topic in \cite{TroelstravanD}, \cite{BishopBridges} as
constructive mathematics) for the details.
\begin{prop}\label{prn:aconverse}
Suppose $(X,d,\nu_0)$ a computable metric space with $\card X\geq 2$,
$\mathcal{Y}\subseteq\{A\in\mathcal{A}(X)\setconstr \dim A\leq 0\}$ is a represented class with
$\delta_{\mathcal{Y}}\leq\delta_{\Pi^0_1(X)}|^{\mathcal{Y}}$, that
$\gamma:X\times\mathcal{Y}\to\mathcal{Y},(x,Y)\mapsto Y\union\{x\}$ is well-defined and computable, and that
\[
\beta:X\times\N\times\mathcal{Y}\to\mathcal{K}(X),(x,n,Y)\mapsto Y\cap\clball{x}{2^{-n}}
\]
is well-defined and $(\delta_X,\delta_{\N},\delta_{\mathcal{Y}};\deltamc)$-computable.
If also $E'$ is $(\delta_5;[\delta_X\to\delta_X]_{\Pi^0_1(X)})$-computable
then
\[
M:\subseteq X\times\Sigma^0_1(X)\times\mathcal{Y}\crsarr\Sigma^0_1(X)^2,
(x,U,Y)\mapsto\{(V,W)\setconstr x\in V\wedge V\subseteq U\wedge Y\subseteq V\disjunion W\}
\]
($\dom M=\{(x,U,Y)\setconstr x\in U\}$) is computable.
\end{prop}
\begin{proof}
Assuming $x\in U$, from inputs
$p\in\delta_{\mathcal{Y}}^{-1}\{Y\}$, $q\in\delta_X^{-1}\{x\}$,
$r\in\delta_{\Sigma^0_1(X)}^{-1}\{U\}$
we can (multi-valuedly in terms of $(Y,x,U)$ but single-valuedly in
terms of $\langle p,q,r\rangle$) computably determine some $n\in\N$
such that $(\exists y\in X)d(x,y)\geq 2^{-n}$,
$\ball{x}{2^{-n}}\subseteq U$ and $\alpha_0,\alpha_1\in\R$ such that
$0<\alpha_0<\alpha_1<2^{-n}$ and
\[
\cl{K\cap\ball{x}{\alpha_j}}=K\cap\clball{x}{\alpha_j}\wedge
K\setminus\ball{x}{\alpha_j}=\cl{K\setminus\clball{x}{\alpha_j}}
\]
for $K=\beta(x,n,Y)$ and both $j<2$; see
\cite[Proof of Prop 8.5]{zerodpaper}.
For $B:=d(x,\cdot)^{-1}[\alpha_0,\alpha_1]$, or rather from a
$\deltamc$-name of $B\cap K=B\cap Y$ (which is computable from the inputs), we can find some
$h\in E'(Y\union\{x\},(B\cap K)\union\{x\})$; note here we use $(B\cap K)\union\{x\}\neq X$ ($\ni y$).  More
precisely, suppose $P:\subseteq\Bairespc\times\Bairespc\to\Bairespc$
is a computable $(\delta,\delta;\delta)$-realizer of binary union on
$\mathcal{A}(X)$, where
$\delta=\deltar\sqcap\delta_{\textrm{dist}}^{>}$, $J,J',J''$ are respective witnesses of $\deltamc|^{\mathcal{K}^{*}(X)}\leq
(\deltar\sqcap\delta_{\textrm{dist}}^{>})|^{\mathcal{K}^{*}(X)}$
(this follows from Thm 4.12, Prop 4.2(1), Thm 4.10(1), Thm 3.9(1) in
\cite{BrattkaPresser}),
$\deltamc\leq\delta_{\Pi^0_1(X)}|^{\mathcal{K}(X)}$ and $\delta_{\mathcal{Y}}\leq\delta_{\Pi^0_1(X)}|^{\mathcal{Y}}$,
$I$ is a computable $(\delta_X,\deltamc)$-realizer of $\iota:X\to
\mathcal{K}(X)$, and
$\tilde{E},G$ are computable realizers of $E'$ and $\gamma$.  Then
letting
$F,F',H,H',\tilde{B},B':\subseteq\Bairespc\times\Bairespc\to\Bairespc$ be computable realizers of:
\begin{enumerate}
\item
$a_{\mathcal{Z}}:\Cont_{\mathcal{Z}}(X,Y)\times\Pi^0_1(Y)\crsarr\Pi^0_1(X)$,
$v_{\mathcal{Z}}:\Cont_{\mathcal{Z}}(X,Y)\times\Sigma^0_1(Y)\crsarr\Sigma^0_1(X)$ from Lemma \ref{lem:preimg} and its
corollary (here with $Y=X$, $\mathcal{Z}=\Pi^0_1(X)$);
\item $\cap:\Pi^0_1(X)^2\to\Pi^0_1(X)$, $\cap:\Sigma^0_1(X)^2\to\Sigma^0_1(X)$;
\item $\subseteq X\times\R\to\Pi^0_1(X),(y,R)\mapsto\clball{y}{R}$ and
$\subseteq X\times\R\to\Sigma^0_1(X),(y,R)\mapsto\ball{y}{R}$,
\end{enumerate}
we find any
\begin{gather*}
\tilde{p}\in\deltamc^{-1}\{B\cap Y\},\onespace
\tilde{q}:=\tilde{E}\langle G\langle q,p\rangle,P(J\tilde{p},(J\compose I)(q))\rangle,\\
\tilde{r}:=H\left(\tilde{B}(q,p^{(1)}),H(F(\tilde{q},(J'\compose I)(q)),(J''\compose G)\langle q,p\rangle)
\right)\:\text{ and }\\
\tilde{s}:=H'\left(B'(q,p^{(0)}),F'(\tilde{q},B'(q,p^{(0)}))\right),
\end{gather*}
where we pick $t^{(j)}\in\rho^{-1}\{\alpha_j\}$ ($j<2$), have
\[
h:=[\delta_X\to\delta_X]_{\Pi^0_1(X)}(\tilde{q}) \in
E'\left( Y\union\{x\},(B\cap Y)\union\{x\} \right).
\]
Also, any $z\in\clball{x}{\alpha_1}\cap h^{-1}\{x\}$ has either
$z\in B\cap h^{-1}\{x\}$ ($=\emptyset$ since $\dom h=Y\union\{x\}$,
$h|_{(Y\union\{x\})\cap B}=h|_{B\cap Y}=\id_{B\cap Y}$ and
$B\not\in x$), or $z\in \ball{x}{\alpha_0}\cap
h^{-1}\ball{x}{\alpha_0}$.  Conversely, if $z$ lies in the latter set,
\(h(z)\in \ball{x}{\alpha_0}\cap\left( (B\cap Y)\union\{x\} \right)
=\{x\}\) so
\(\ball{x}{\alpha_0}\cap h^{-1}\ball{x}{\alpha_0}\subseteq
\clball{x}{\alpha_1}\cap h^{-1}\{x\}\).  Thus
\[
D:=\clball{x}{\alpha_1}\cap h^{-1}\{x\}=\ball{x}{\alpha_0}\cap
h^{-1}\ball{x}{\alpha_0}
=\delta_{\Pi^0_1}(\tilde{r})=\delta_{\Sigma^0_1}(\tilde{s})\cap\dom h
\in\Delta^0_1(Y\union\{x\}).
\]
In particular,
$Y\union\{x\}\subseteq\delta_{\Sigma^0_1}(\tilde{r})\union
\delta_{\Pi^0_1}(\tilde{r}) \subseteq
\delta_{\Sigma^0_1}(\tilde{r})\union\delta_{\Sigma^0_1}(\tilde{s})$,
and if
$C:=(Y\union\{x\})\setminus D$ we have $C\cap D=\emptyset$ and
$D\subseteq\clball{x}{\alpha_1}\subseteq U$.  Then picking
$V,W\in\Sigma^0_1(X)$ such that $(V,W)\in t_4(D,C)$ (by computability of $t_4$), we can define $V':=V\cap U$ and get \(x\in D\subseteq
V'\subseteq U\wedge Y\union\{x\}=D\union C\subseteq V\disjunion W\).  But this shows $M$ is computable.
\end{proof}
Combining Theorem \ref{thm:a}, Proposition \ref{prn:aconverse} and results in Section \ref{sec:zdsubsets}, we
have the following
\begin{cor}
If $X$ is a computable metric space, \(\mathcal{Y}\subseteq
\{A\in\mathcal{A}(X)\setconstr\dim A\leq 0\}\) with
$\delta_{\mathcal{Y}}\leq\delta_{\Pi^0_1(X)}|^{\mathcal{Y}}$,
$\gamma:X\times\mathcal{Y}\to\mathcal{Y},(x,Y)\mapsto Y\union\{x\}$
well-defined and computable, and $\beta$ well-defined and computable,
then $E'$ computable iff $(\mathcal{Y},\delta_{\mathcal{Y}})$ uniformly effectively zero-dimensional.
\end{cor}
\section{An analogue in nonarchimedean analysis}\label{sec:fourth}
In this section we establish the result (Theorem \ref{thm:b}) on computing a retraction given an arbitrary
Dugundji system (and some additional data), mentioned in the introduction.  We first introduce concepts of
nonarchimedean analysis.  Several representations of the $p$-adic numbers
$\Omega_p$ ($p$ a prime) are introduced and studied already in \cite{KapoulasThesis}, in particular with complexity considerations in mind; though we do not make explicit use of that material we will use the Cauchy
representation of $\Omega_p$ obtained by introducing a computable metric structure compatible with the field
operations.  More precisely recall (from \cite{Schikhof}) a valued field is a pair $(K,\absval{\cdot})$ where $K$
is a field and $\absval{\cdot}:K\to\R$ is a \emph{valuation}, i.e.~satisfying $\absval{x}\geq 0$,
$(\absval{x}=0\iff x=0_K)$ for all $x\in K$, $\absval{x+y}\leq\absval{x}+\absval{y}$ and
$\absval{x y}=\absval{x}.\absval{y}$ for all $x,y\in K$.  Any valuation induces a metric on $K$ by
$d(x,y):=\absval{x-y}$.  In particular, $\Omega_p$ is here introduced as the completion of the metric space
given by the following valuation $\absval{\cdot}_p$ on $\Q$:
\[
\absval{x}_p:=\begin{cases}
p^{-n}, &\mbox{\textrm{ if $(\exists n\in\Z)(\exists s,t\in\Z\setminus\{0\})(p\not\;\divides s\wedge
p\not\;\divides t\wedge x=p^n\frac{s}{t})$,}}\\
0,&\mbox{\textrm{ if $x=0$}}.
\end{cases}
\]
As such, a standard numbering $\nu_{\Q}:\N\to\Q\subseteq\Omega_p$ serves as the sequence of our computable metric structure (one checks
$\Q\times\Q\to\Q,(r,s)\mapsto\absval{r-s}_p$ is well-defined and $([\nu_{\Q},\nu_{\Q}],\nu_{\Q})$-computable).  From now on we will consider
primarily $K=\Omega_p$
with the Cauchy representation $\delta_K$.  Recall or note that in fact $\absval{x+y}\leq\max\{\absval{x},\absval{y}\}$ for all $x,y\in K$
(strong triangle inequality); we say that $K$ is a \emph{non-archimedean} valued field.  We now briefly go over the construction of computable
realizers for the field operations in $K$.  Here we denote $K^{\times}:=K\setminus \{0_K\}$ and fix some $e_0\in\nu_{\Q}^{-1}\{0\}$,
$e_1\in\nu_{\Q}^{-1}\{1\}$.
\begin{prop}
The field operations on $K=\Omega_p$ are computable with respect to
$\delta_K$.  Also, $\iota:\N\to K$ is
$(\delta_{\N},\delta_K)$-computable and $\absval{\cdot}:K\to\R$ is $(\delta_K,\rho)$-computable.
\end{prop}
\begin{proof}
The proof loosely follows \cite[Ex 1.B]{Schikhof}.  Firstly $-:K\to K$ is $1$-Lipschitz, and observe
$-(\Q)\subseteq\Q$ with $-:\Q\to\Q$ being $(\nu_{\Q},\nu_{\Q})$-computable.  So $-:K\to K$ is computable.
Similarly $+:K\times K\to K$ is $1$-Lipschitz using the maximum metric on $K\times K$ (and the strong triangle
inequality), while $\Q+\Q\subseteq\Q$ and $+:\Q\times\Q\to\Q$ is $([\nu_{\Q},\nu_{\Q}],\nu_{\Q})$-computable,
so $+:K\times K\to K$ is computable (recall also that $\delta_{X\times X}\equiv[\delta_X,\delta_X]$).  If
$g\in\trfns$ is a $(\nu_{\Q},\nu_{\Q})$-realizer of $\absval{\cdot}_p|_{\Q}$ (using the formula above), for
$q\in\delta_K^{-1}\{x\}$ and $i\in\N$ we can note
\[
2^{-i}\geq\absval{x-\nu_{\Q}(q_i)}_p \geq
\absval{\absval{x}_p-\absval{\nu_{\Q}(q_i)}_p}=\absval{\absval{x}_p-(\nu_{\Q}\compose g)(q_i)}.
\]
Then $F:\subseteq\Bairespc\times\N\to\N,(q,i)\mapsto g(q_i)$ witnesses computability of
$\absval{\cdot}:K\to\R$ according to the following
\begin{lem}\label{lem:trivlemone}
If $(X,\delta_X)$ is a represented space and $(Y,d,\lambda)$ a computable metric space with
$f:\subseteq X\to Y$ having some computable $F:\subseteq\Bairespc\times\N\to\N$ such that
\[
(\forall p\in\delta_X^{-1}\dom f)(\forall i)d((f\compose\delta_X)(p),(\lambda\compose F)(p,i))<2^{-i},
\]
then $f$ is $(\delta_X,\delta_Y)$-computable where $\delta_Y$ is the Cauchy representation.
\end{lem}
To apply Lemma \ref{lem:trivlemone} to see $\cdot:K\times K\to K$ is computable (assuming $\rho$ is the usual
Cauchy representation of $\R$), on input $\langle q,r\rangle\in\dom[\delta_K,\delta_K]$ and $i\in\N$ we
compute $\delta,\eta\in\Qplus$ such that
\begin{equation}\label{eq:multeps}
\absval{y}.\delta<\epsilon \wedge (\absval{x}+\delta).\eta<\epsilon
\end{equation}
for $\epsilon=2^{-i}$, $x=\delta_K(q)$, $y=\delta_K(r)$; then $k,l\in\N$ such that $2^{-k}\leq\delta$ and
$2^{-l}\leq\eta$; then finally pick $F(\langle q,r\rangle,i)\in \nu_{\Q}^{-1}\{\nu_{\Q}(q_k).\nu_{\Q}(r_l)\}$.
Here (\ref{eq:multeps}) ensures any $z\in\ball{x}{\delta}$,
$w\in\ball{y}{\eta}$ have
\[
\absval{x y-z w}\leq\max\{\absval{x y-z y},\absval{z y-z w}\} \leq
\max\{\absval{y}.\delta,\absval{z}.\eta\}<\epsilon
\]
(since $\absval{z}\leq\absval{x}+\delta$).  Testing the inequalities can be made effective by using
(if $V:\subseteq\Bairespc\to\Bairespc$ is a computable realizer of
$\absval{\cdot}:K\to\R$) the equivalence
\[
\absval{y}.\delta=(\rho\compose V)(r).\delta<\epsilon \iff
(\exists m)(\nu_{\Q}(V(r)_m)+2^{-m}).\delta<\epsilon,
\]
similarly for the second inequality.

Similarly, to see $\cdot^{-1}:\subseteq K\to K$ is computable (with natural domain), we are instead given
$q\in\delta_K^{-1}K^{\times}$ and $i\in\N$, so compute $\delta\in\Qplus$ such that
\begin{equation}\label{eq:inveps}
\delta<\absval{x}\wedge \frac{\delta}{\absval{x}-\delta}<\epsilon.\absval{x}
\end{equation}
for $\epsilon=2^{-i}$, $x=\delta_K(q)$; then compute $k\in\N$ such that $2^{-k}\leq\delta$ and let
$F(q,i)\in\nu_{\Q}^{-1}\{\nu_{\Q}(q_k)^{-1}\}$.  Note that (\ref{eq:inveps}) ensures any $y\in\ball{x}{\delta}$
has $\absval{y}\geq\absval{\absval{x}-\absval{x-y}}$, so
\[
\bigg\lvert\frac{1}{x}-\frac{1}{y}\bigg\rvert
=\frac{\absval{y-x}}{\absval{x}\absval{y}} <
\frac{\delta}{\absval{x}.\absval{y}} \leq
\frac{\delta}{\absval{x}.\absval{\absval{x}-\absval{x-y}}}=\frac{\delta}{\absval{x}.(\absval{x}-\absval{x-y})}
<\frac{\delta}{\absval{x}.(\absval{x}-\delta)}<\epsilon
\]
where we used $\absval{x}-\absval{x-y} > \absval{x}-\delta>0$.

Finally, note $e_0^{\ordomega}\in\delta_K^{-1}\{0_K\}$, $e_1^{\ordomega}\in\delta_K^{-1}\{1_K\}$, then define
$J(0.p):=e_0^{\ordomega}$, $J(1.p):=e_1^{\ordomega}$ and inductively $J((n+2).p):=P(J((n+1).p),J(1.p))$
($n\in\N$).  The function $J:\subseteq\Bairespc\to\Bairespc$ thus
defined is plainly a computable $(\delta_{\N},\delta_K)$-realizer of $\iota:\N\to K$, provided
$P:\subseteq\Bairespc^2\to\Bairespc$ is a computable realizer of $+:K\times K\to K$.
\end{proof}

Next, we follow \cite{Schikhof} in introducing ultrametrically convex sets.  In any ultrametric space $(X,d)$,
denote $[x,y]:=\clball{x}{d(x,y)}=\clball{y}{d(x,y)}$; this is the unique smallest-radius closed ball that
contains $\{x,y\}$.
\begin{proof}[Proof sketch]
We know by \cite[Prop 18.5]{Schikhof} that any closed balls $B_1,B_2$ with $B_1\cap B_2\neq\emptyset$ are
$\subseteq$-comparable.  If $B_1:=\clball{x}{d(x,y)}$, $B_2:=\clball{y}{d(x,y)}$ and $B_3:=\clball{z}{r}$ with
$\{x,y\}\subseteq B_3\subseteq B_1\cap B_2$, say without loss of generality
$B_1\subseteq B_2$, then any $w\in B_2$ has
\[
d(w,z)\leq\max\{d(w,y),d(y,z)\}\leq d(x,y)\leq\max\{d(x,z),d(z,y)\}\leq r
\]
so $B_3=B_1=B_2$ and $r\geq d(x,y)$.
\end{proof}
\begin{lem}(e.g.~\cite[Prop 24.2(i)]{Schikhof})
If $K$ is a non-archimedean valued field and $x,y,z\in K$ then $z\in [x,y]$ iff
$(\exists\lambda\in K)(\absval{\lambda}\leq 1\wedge z=\lambda x+(1-\lambda)y)$.
\end{lem}
\begin{proof}
If $\lambda\in K$ with $\absval{\lambda}\leq 1$ then $z:=\lambda x+(1-\lambda)y$ has
\[
\absval{z-y}=\absval{\lambda(x-y)}=\absval{\lambda}.\absval{x-y}\leq\absval{x-y} \implies
z\in\clball{y}{\absval{x-y}}=[x,y].
\]
Conversely, given
\[
z\in[x,y]=\clball{x}{\absval{x-y}}=\clball{y}{\absval{x-y}}
\]
where $x\neq y$, note $\lambda:=\frac{z-y}{x-y}$ has $1-\lambda=\frac{(x-y)-(z-y)}{x-y}=\frac{x-z}{x-y}$ and
$\absval{\lambda}=\frac{\absval{z-y}}{\absval{x-y}}\leq 1$ (the equality follows from $\lambda(x-y)=z-y$ and
the last property of valuations, plus the fact $K$ is a field with $x\neq y$).  Now
\[
\lambda x+(1-\lambda)y =
\frac{z-y}{x-y}z+\frac{x-z}{x-y}y=\frac{(z-y)x+(x-z)y}{x-y}=\frac{z(x-y)}{x-y}=z.
\]
\end{proof}
Consequently, by loose analogy with $\R$-vector spaces, a subset $C\subseteq X$ of ultrametric space $(X,d)$ is
\emph{ultrametrically convex} if $(\forall x,y\in C)([x,y]\subseteq C)$.
\begin{lem}(cf.~\cite[Prop 24.2(iv)]{Schikhof})
Fix an integer $n\geq 2$ and $C\subseteq K$ for a nonarchimedean valued field $K$.  Then $C$ satisfies
\[
\Cvx_n(C):\equiv (\forall w\in C^n)(\forall\alpha\in K^n) \left(
(\forall i\in[n])\absval{\alpha_i}\leq 1\wedge\sum_{i=1}^n\alpha_i=1_K\implies\sum_{i=1}^n\alpha_i w_i\in C
\right)
\]
iff it has the form $C=\{x\in K\setconstr\absval{x-a}\leq r\}$ or $C=\{x\in K\setconstr\absval{x-a}<r\}$ for
some $r\in [0,\infty]$, $a\in K$.
\end{lem}
\begin{proof}
\textbf{($\mathbf{\implies}$):} First notice that $\Cvx_m(C)$ implies $\Cvx_n(C)$ whenever $m\geq n$ (since we
can set other parameters $\alpha_i$ equal to $0_K$).  Suppose then that $C\subseteq K$ is ultrametrically
convex and $x\in C$.  We know
\[
\ball{x}{r}\subseteq\bigcup\{[x,y]\setconstr y\in C\}\subseteq
\clball{x}{r}\cap C
\]
for $r:=\sup\{\absval{x-y}\setconstr y\in C\}$, where we allow zero or infinite radius
in the open and closed balls.  Of course, $C\subseteq \bigcup\{[x,y]\setconstr y\in C\}$ here, and if
$C\neq\ball{x}{r}$ there must exist $y\in C$ such that
$\absval{x-y}=r$, and then $[x,y]=\clball{x}{r}\subseteq C$,
so $C=\clball{x}{r}$.\\
\textbf{($\mathbf{\impliedby}$):} It is easy to see $C:=\clball{0_K}{r}$ satisfies $\Cvx_n(C)$: if
$w_1,\dots,w_n\in C$ and $\alpha_1,\dots,\alpha_n\in K$ with
$(\forall i)\absval{\alpha_i}\leq 1\wedge\sum_i\alpha_i=1_K$ then $y:=\sum_i\alpha_i w_i$ has
\[
\absval{y}\leq \max_{i\in[n]}\absval{\alpha_i w_i}\leq\max_{i\in[n]}\absval{w_i}\leq r
\]
by repeated use of
the strong triangle inequality.  Similarly the open ball $C=\ball{0_K}{r}$ satisfies $\Cvx_n(C)$.  We now
verify that $(\Cvx_n(C)\implies\Cvx_n(a+C))$ for any $a\in K$: if $w_1,\dots,w_n\in a+C$ and
$\alpha_1,\dots,\alpha_n\in K$ with $(\forall i)\absval{\alpha_i}\leq 1\wedge\sum_i\alpha_i=1_K$ then
$\sum_i\alpha_i w_i=\sum_i\alpha_i(w_i-a)+a\in a+C$.
\end{proof}

Next, we deal computationally with partitions of unity in a nonarchimedean context.  We will consider locally finite `covers' $(V_i)_i\in\textrm{LF}'_{\Sigma^0_1(X),\N}$ with
nonarchimedean analogues of partitions of unity $(f_i)_{i\in\N}$ subordinate to $(V_i)_i$.  More formally,
\begin{defi}
A \emph{partition of unity} subordinate to $(V_i)_i\subseteq\Sigma^0_1(X)$ with
values in $K$ is a sequence $(f_i)_i\in\Cont_{\Sigma^0_1(X)}(X,K)^{\N}$ with
$\dom f_i=Y:=\bigcup_j V_j$ and $X\setminus f_i^{-1}\{0_K\}\subseteq V_i$ for each $i\in\N$, and with
$(\forall i)(\forall x\in\dom f_i)(\absval{f_i(x)}\leq 1)$ and $(\forall x\in Y)(\sum_{i\in\N}f_i(x)=1_K)$.
\end{defi}
Note the sum here is locally finite by the local finiteness of $(V_i)_i$.

Now, suppose $X$ is an effectively zero-dimensional computable metric space.  Suppose also we are given a
subset $A\in\mathcal{A}(X)\setminus\{\emptyset\}$, Dugundji system $((V_i)_i,(y_i)_i)$ for $A$
and partition of unity $(f_i)_i$ subordinate to $(V_i)_i$ with values in $K=\Omega_p$.  If there also exists a
homeomorphism $h:A\to C$ for some ultrametrically convex $C\subseteq K$ then one can attempt to generalize the
above Theorem \ref{thm:a} in the following way: let
\begin{equation}\label{eq:genretract}
r(x)=\begin{cases}
x, &\mbox{\textrm{ if $x\in A$,}}\\
h^{-1}\left( \sum_{i\in\N}h(y_i).f_i(x) \right),&\mbox{\textrm{ if $x\in X\setminus A$}}.
\end{cases}
\end{equation}
This should define a retraction $r:X\to X$ onto $A$, as we will see presently.
The existence of such a homeomorphism $h$, though, is rather a restricted phenomenon, requiring that $A$
have either no isolated points or else is a singleton.

\begin{thm}\label{thm:b}
For a zero-dimensional computable metric space $X$ define
\begin{gather*}
\theta:\subseteq\mathcal{A}(X)\times\mathcal{D}(X) 
\times\Cont_{\Sigma^0_1(X)}(X,K)^{\N}\times\Cont_{\Pi^0_1(X)}(X,K)\to\Cont(X,K)\\
\text{by }\onespace\theta(A;(V_i)_i,(y_i)_i;(f_i)_i;h):=h\compose r
\end{gather*}
where $r$ is as defined in (\ref{eq:genretract}) and $\dom\theta$ consists of all tuples
$(A;(V_i)_i,(y_i)_i;(f_i)_i;h)$ such that:
\begin{enumerate}
\item
$\emptyset\neq A\neq X$ and $((V_i)_i,(y_i)_i)\in \pi^{-1}\{A\}$,
\item
$(f_i)_i$ is a partition of unity subordinate to $(V_i)_i$ with values in $K$, and
\item $h:A\to h(A)\subseteq K$ is a homeomorphism onto an ultrametrically convex set in $K$.
\end{enumerate}
If
$\beta:X\times\N\to\mathcal{K}_{>}(X),(x,N)\mapsto\clball{x}{2^{-N}}$ is well-defined and computable, then
$\theta$ is $(\delta,[\delta_X\to\delta_K])$-computable where
$\delta:=[\deltar\sqcap\delta_{\textrm{dist}}^{>},\delta_3,
[\delta_X\to\delta_K]_{\Sigma^0_1(X)}^{\ordomega},
[\delta_X\to\delta_K]_{\Pi^0_1(X)}]$.
\end{thm}
\begin{rem}
For examples of maps $h$ as in
(3), one can suppose $A\in\mathcal{K}(X)\setminus\{\emptyset,X\}$ is
perfect and let $C:=\clballmetric{K}{0_K}{\rho}$ for some
$0<\rho<\infty$.  Then $C$ is compact, perfect and zero-dimensional,
so a homeomorphism $h:A\to C$ is well-known to exist (e.g.~it follows
from \cite[Ex 6.2.A]{Engelking}, \cite[Prop 74.2]{Schikhof} or
\cite[Thm (7.4)]{Kechris}).  Information on $A$ and $C$ sufficient to
compute a $[\delta_X\to\delta_K]_{\Pi^0_1(X)}$-name of some such $h$
is of interest in this context; one simple result in this
direction follows \cite[Proof of Thm (7.4)]{Kechris} to decompose
$A$ and $C$ into (nonempty) relatively clopen sets according to Cantor
schemes $(A_u)_{u\in\{0,1\}^{*}}$ and $(B_u)_u$, then take associated
maps $f:\{0,1\}^{\N}\to A$ and $f':\{0,1\}^{\N}\to C$.  The composition
$h:=f'\compose f^{-1}:A\to C$ is then computable uniformly from a
\([\nu_{\{0,1\}^{*}}\to \delta_{\Sigma^0_1(A)}]\)-name of $(A_u)_u$, a
\([\nu_{\{0,1\}^{*}}\to \deltamc^K]\)-name of $(B_u)_u$ and a
$\delta_{\Pi^0_1(X)}$-name of $A$.

The condition (2) can also be met by picking clopen sets $W_{i,j}$ ($j\in\N$) decomposing each $V_i$, with each $f_i:\subseteq X\to K$ having domain $X\setminus A$ and constant on each $W_{i,j}$.  Using
Remark \ref{rem:zddgone} one can replace the locally finite cover $(V_i)_i$ here by $(W_{i,j})_{\langle i,j\rangle}$ if desired.
\end{rem}
\begin{proof}[Proof of Theorem \ref{thm:b}]
By type conversion it is enough to compute the map $\tilde{\theta}:Z\times X\to K,(z,x)\mapsto\theta(z)(x)$
where $Z:=\dom\theta$ is represented by $\delta|^Z$.  To do this, we loosely follow the proof of
Theorem \ref{thm:a}, computing $(h\compose r)^{-1}U\in\Sigma^0_1(X)$ uniformly from $U\in\Sigma^0_1(K)$.  We
will take $q\in\delta_{\Sigma^0_1(K)}^{-1}\{U\}$ and hope to define computable
$F:\subseteq\Bairespc^3\to\Bairespc,(p,\hat{p},q)\mapsto t$ such that each $\hat{p}\in\dom\delta|^Z$ and
$q\in\Bairespc$ have corresponding
$u=(\delta_{\Sigma^0_1(X)}\compose F)(\cdot,\hat{p},q):\dom\delta_X\to\Sigma^0_1(X)$ satisfying
\begin{align*}
\delta_X(p)\in u(p)&\subseteq (h\compose r)^{-1}\delta_{\Sigma^0_1(K)}(q)
\text{ if }p\in\delta_X^{-1}(h\compose r)^{-1}\delta_{\Sigma^0_1(K)}(q);\\
u(p)&=\emptyset\text{ if }p\in\dom\delta_X\setminus\delta_X^{-1}(h\compose r)^{-1}\delta_{\Sigma^0_1(K)}(q).
\end{align*}
Then Lemma \ref{lem:ctslindelof} will computably give a name for $(h\compose r)^{-1}\delta_{\Sigma^0_1(K)}(q)$.

To define $t$, we will dovetail repeated output of $0$ with testing
\[
R_0(p,\hat{p},q,j)\vee
R_1(p,\hat{p},q,j,M,n,N,\langle w\rangle)
\]
for all $j,M,n,N\in\N$, $w\in\N^{*}$, instead outputting
$(a+1).0^{\ordomega}$ if such $j,M,n,N,w$ are found, where $a:=\langle p_j,\overline{2^{-j+1}}\rangle$.
This time, if
%
$\hat{p}=\langle s,\langle\tilde{d},\tilde{t}\rangle,\tilde{f},e\rangle$
where $\tilde{d}=\langle\langle q^{(0)},\dots\rangle,\tilde{q},s'\rangle$,
$\tilde{t}=\langle t^{(0)},\dots\rangle$ and $\tilde{f}=\langle f^{(0)},\dots\rangle$, we define
\begin{align*}
R_0(p,\hat{p},q,j) &:\equiv (\exists k)\tilde{q}_k\geq 1\wedge\langle p_j,\overline{2^{-j+1}}\rangle\sqsubset
\tilde{q}_k-1\wedge (\exists c,l)\langle p_j,\overline{2^{-j+1}}\rangle\sqsubset c\wedge\\
& c+1=V(G\langle\hat{p},k.0^{\ordomega}\rangle,q)_l
\end{align*}
where $V:\subseteq\Bairespc^2\to\Bairespc$ is a computable
$([\delta_X\to\delta_K]_{\Sigma^0_1(X)},\delta_{\Sigma^0_1(K)};\delta_{\Sigma^0_1(X)})$-realizer of
$v:\Cont_{\Sigma^0_1(X)}(X,K)\times\Sigma^0_1(K)\to\Sigma^0_1(X),(g,U)\mapsto g^{-1}U$ (compare
Lemma \ref{lem:preimg}), and $G$ is a computable
$([\delta',\delta_{\N}];[\delta_X\to\delta_K]_{\Sigma^0_1(X)})$-realizer of
\[
\Gamma:\check{Z}\times\N\to\Cont_{\Sigma^0_1(X)}(X,K),
(A;\langle\tilde{d},\tilde{t}\rangle;(f_i)_i;h;k)\mapsto g_k,
\]
for $g_k:\subseteq X\to K$ defined by $\dom g_k=X\setminus A$ and
$g_k(x):=\sum_{j\in \FS(\id_{\N})(s_k')}h(y_j).f_j(x)$
($x\in X\setminus A$).  Here we wrote $\tilde{d}$ in the form given before (to get $s'$), and define
$\check{Z}$ to be
\[
\{(\deltar\sqcap\delta_{\textrm{dist}}^{>}(s);\langle\tilde{d},\tilde{t}\rangle;
[\delta_X\to\delta_K]_{\Sigma^0_1(X)}^{\ordomega}(\tilde{f});[\delta_X\to\delta_K]_{\Pi^0_1(X)}(e))\setconstr
\hat{p}=\langle s,\langle\tilde{d},\tilde{t}\rangle,\tilde{f},e\rangle\in\delta^{-1}Z\},
\]
with representation $\delta'$ of $\check{Z}$ given by
\[
\delta'(\hat{p}):=(\deltar\sqcap\delta_{\textrm{dist}}^{>}(s);
\langle\tilde{d},\tilde{t}\rangle;[\delta_X\to\delta_K]_{\Sigma^0_1(X)}^{\ordomega}(\tilde{f});
[\delta_X\to\delta_K]_{\Pi^0_1(X)}(e)).
\]

Also, $R_1$ is defined by
\[
\begin{split}
R_1(p,\hat{p},q,j,M,n,N,\langle w\rangle) &:\equiv d(\nu(p_j),z_n)+2^{-j+1}<\frac{1}{M+1}\wedge
\frac{3}{M+1}\leq 2^{-N}\wedge \\
& (\exists m')\langle w\rangle =
I(\langle q^{(0)},\dots\rangle,\tilde{B}(P s^{(n)},N.0^{\ordomega}))_{m'}\wedge\\
& (\exists l)(\forall i<\absval{w}) \big(
(\exists m)\left( w_i+1=V'\langle e,H(\langle n,l\rangle.0^{\ordomega},s,e)\rangle_m \right)\wedge\\
& (\exists d,m)\absval{h(z_n)-\nu_K(\pr_1 d)}_p+2^{-l}<\nu_{\Qplus}(\pr_2 d)\wedge q_m=d+1 \big),
\end{split}
\]
where $I:\subseteq\Bairespc^2\to\Bairespc$ is a computable realizer of
$\cap:\Pi^0_1(X)\times\mathcal{K}_{>}(X)\to\mathcal{K}_{>}(X)$, $\tilde{B}:\subseteq\Bairespc^2\to\Bairespc$ is a
computable $(\delta_X,\delta_{\N};\deltac)$-realizer of $\beta$, $H:\subseteq\Bairespc^3\to\Bairespc$ is a
computable $(\delta_{\N},\id_{\Bairespc},\id_{\Bairespc};\delta_{\Sigma^0_1(K)})$-realizer of
$\subseteq\N\times\Bairespc^2\to\Sigma^0_1(K),(\langle n,l\rangle,s,e)
\mapsto \ball{h(z_n)}{2^{-l}}$ (with natural domain) for $h=[\delta_X\to\delta_K]_{\Pi^0_1(X)}(e)$,
$\pr_1 s=\langle s^{(0)},\dots\rangle$, $z_n=\delta_X(P s^{(n)})$ (and $P:\subseteq\Bairespc\to\Bairespc$ as in
definition of $\deltar$), and where $V':\subseteq\Bairespc\to\Bairespc$ is a computable realizer of
$v_{\Pi^0_1(X)}$ for $v_{\mathcal{Z}}:\Cont_{\mathcal{Z}}(X,K)\times\Sigma^0_1(K)\crsarr\Sigma^0_1(X),
(f,U)\mapsto\{V\setconstr V\cap\dom f=f^{-1}U\}$ as in Lemma \ref{lem:preimg}.

Suppose now $x\in X$, $p\in\delta_X^{-1}\{x\}$, $\hat{p}\in\delta^{-1}Z$, $q\in\Bairespc$.
Any nonzero output $a+1$ (on input $p,\hat{p},q$) must have the form
$a=\langle p_j,\overline{2^{-j+1}}\rangle$, and we show any $z\in\alpha(a)$ has $(h\compose r)(z)\in U$.
Namely, either:
\begin{enumerate}
\item[(I)] (if $R_0(p,\hat{p},q,j)$ holds)\\
we have for any $k$ such that
$\alpha'(\tilde{q}_k)\ni z$ --- in particular the value of $k$ found by the test ---
that $z\in\alpha'(\tilde{q}_k)\subseteq X\setminus A$ and
\[
(h\compose r)(z)=\sum_{i\in\N}h(y_i).f_i(z)
=\sum_{i\in\FS(\id_{\N})(s'_k)}h(y_i).f_i(z)\in U,
\]
since
\begin{gather*}
\{i\setconstr\alpha'(\tilde{q}_k)\setminus f_i^{-1}\{0_K\}\neq\emptyset\}\subseteq
\{i\setconstr V_i\cap\alpha'(\tilde{q}_k)\neq\emptyset\}\subseteq \FS(\id_{\N})(s'_k)
\onespace\text{ and }\onespace\\
z\in\alpha(a)\subseteq (\delta_{\Sigma^0_1(X)}\compose V)(G\langle\hat{p},k.0^{\ordomega}\rangle,q) =
g_k^{-1}\delta_{\Sigma^0_1(K)}(q),
\end{gather*}
or else
\item[(II)] (if $R_1(p,\hat{p},q,j,M,n,N,\langle w\rangle)$ holds)\\
we can argue as follows.  Either $z\in A$ or
$(\exists k)\alpha'(\tilde{q}_k)\ni z$.  In the first case, clearly
\[
z\in\alpha(a)\cap A\subseteq \ball{z_n}{\frac{1}{M+1}}\cap A\subseteq\clball{z_n}{2^{-N}}\cap A\subseteq
h^{-1}\ball{h(z_n)}{2^{-l}}\subseteq h^{-1}U,
\]
implying $(h\compose r)(z)=h(z)\in U$.  In the second case, instead $y=z_n\in A$ has
\[
\absval{(h\compose r)(y)-(h\compose r)(z)}_p=\big\lvert\sum_{i\in S}h(y).f_i(z)-h(y_i).f_i(z)\big\rvert_p
\leq \max_{i\in S}\absval{f_i(z)}_p.\absval{h(y)-h(y_i)}_p
\]
for $S:=\FS(\id_{\N})(s_k')$ ($\supseteq
\{i\setconstr \alpha'(\tilde{q}_k)\setminus f_i^{-1}\{0_K\}\neq\emptyset\}$).
(Our argument here derives from the argument for ordinary partitions of unity suggested in
\cite[Prob 4.5.20(a)]{Engelking}).
But here any $i\in S$ with $\absval{f_i(z)}_p\neq 0$ has $V_i\ni z$ so $d(z,y_i)\leq (1+\epsilon)d_A(z)
\leq (1+\epsilon)d(z,y)$,
implying $d(y,y_i)\leq (2+\epsilon)d(z,y)<\frac{2+\epsilon}{M+1}\leq 2^{-N}$ (by $\epsilon\leq 1$ and first
and second clauses of $R_1(p,\hat{p},q,j,M,n,N,\langle w\rangle)$) and hence
$h(y_i)\in\ball{h(z_n)}{2^{-l}}$ (by third and fourth clauses of $R_1(p,\hat{p},q,j,M,n,N,\langle w\rangle)$).
Then we get $\absval{(h\compose r)(y)-(h\compose r)(z)}_p<2^{-l}\implies z\in (h\compose r)^{-1}U$ (since
by the second half of the last clause of
$R_1(p,\hat{p},q,j,M,n,N,\langle w\rangle)$
we have some $d\in\N$ such that $\ballmetric{\absval{\cdot}_p}{h(z_n)}{2^{-l}}\subseteq\alpha^K(d)\subseteq U$).
\end{enumerate}
So, we have proved $\delta_{\Sigma^0_1(X)}(t)=\alpha(a)\subseteq(h\compose r)^{-1}U$.  If
$x\not\in(h\compose r)^{-1}U$, we must have output $t=0^{\ordomega}$ to avoid a contradiction (since $x\in
\alpha\langle p_j,\overline{2^{-j+1}}\rangle$ for all $j\in\N$).
Conversely, for any $x\in(h\compose r)^{-1}U$ we have either $x\in A$ or $(\exists k)\alpha'(\tilde{q}_k)\ni x$.
\begin{enumerate}
\item[(A)] $x\in A$\\
In the former case, $U\ni(h\compose r)(x)=h(x)$ and we can pick $l$ such that
\[
(\exists d,m)\absval{h(x)-\nu_K(\pr_1 d)}_p+2^{-l}<\nu_{\Qplus}(\pr_2 d)\wedge q_m=d+1;
\]
then
$\xi\in\Qplus$ such that $\clball{x}{\xi}\cap A\subseteq h^{-1}\ball{h(x)}{2^{-l-1}}$; then $M,N$ such that
$\frac{3}{M+1}\leq 2^{-N}\wedge\frac{1}{M+1}+2^{-N}\leq\xi$; then $n$ such that
$\absval{h(z_n)-\nu_K(\pr_1 d)}_p+2^{-l}<\nu_{\Qplus}(\pr_2 d)$, $d(x,z_n)<(M+1)^{-1}$ and
$\absval{h x-h z_n}_p<2^{-l-1}$; then finally $j$ such that
$d(\nu(p_j),z_n)+2^{-j+1}<\frac{1}{M+1}$.  Here we find
\begin{equation}\label{eq:cveq}
\clball{z_n}{2^{-N}}\cap A\subseteq h^{-1}\ball{h(z_n)}{2^{-l}}
\end{equation}
since any point $y$ of the left hand side has $y\in\ball{x}{2^{-N}+\frac{1}{M+1}}\cap A$ so
\[
\absval{h y-h z_n}_p\leq \absval{h y-h x}_p+\absval{h x-h z_n}_p<2^{-l-1}+2^{-l-1}.
\]
Now (\ref{eq:cveq}) and definition of $I$, $\tilde{B}$
imply $\clball{z_n}{2^{-N}}\cap A$ has an ideal cover $w\in\N^{*}$ with
\begin{gather*}
(\forall i<\absval{w})(\exists m)\left(w_i+1=V'\langle e,H(\langle n,l\rangle.0^{\ordomega},s,e)\rangle_m
\right)\wedge \\
(\exists m')\langle w\rangle=I(\langle q^{(0)},\dots\rangle,\tilde{B}(s^{(n)},N.0^{\ordomega}))_{m'}.
\end{gather*}
This guarantees that $R_1(p,\hat{p},q,j,M,n,N,\langle w\rangle)$ holds, and $x\in\alpha(a)$ for some nonzero
output $a+1$.
\item[(B)] If instead $\alpha'(\tilde{q}_k)\ni x$ for some $k$,\\
we have
\[
U\ni(h\compose r)(x)=\sum_{i\in\N}h(y_i).f_i(x)=g_k(x),
\]
so we know there exist $c,l$ such that
$c+1=V(G\langle\hat{p},k.0^{\ordomega}\rangle,q)_l\wedge x\in\alpha(c)$.  Then by continuity of $d$ and
$\lim_{j\to\infty}2^{-j+1}=0$ we have some $j$ such that
the first, second and third clauses of $R_0(p,\hat{p},q,j)$
hold, so that $R_0(p,\hat{p},q,j)$ holds.  But then $a+1$ appears in the output where
$a=\langle p_j,\overline{2^{-j+1}}\rangle$, in particular
$x\in\alpha(a)\subseteq(\delta_{\Sigma^0_1(X)}\compose F)(p,\hat{p},q)$.
\end{enumerate}
Thus we have shown, given $\hat{p}\in(\delta')^{-1}\check{Z}$, $q\in\delta_{\Sigma^0_1(X)}^{-1}\{U\}$,
$p\in\delta_X^{-1}\{x\}$ where $x\in(h\compose r)^{-1}U$, that
$(h\compose r)^{-1}U\subseteq \bigcup\{(\delta_{\Sigma^0_1(X)}\compose F)(p,\hat{p},q)\setconstr
p\in\delta_X^{-1}(h\compose r)^{-1}\delta_{\Sigma^0_1(K)}(q)\}$.  This completes the proof.
\end{proof}

\section{Zero dimensional subsets}\label{sec:zdsubsets}
In \cite[\S 4]{zerodpaper}, the present author introduced four `natural' operations $N,M,B,S$ concerning a general represented class $\mathcal{Y}\subseteq\{A\subseteq X\setconstr \dim A\leq 0\}$, finding that $N$ is computable only if $M$ is computable, $M$ computable only if $B$ computable, and $B$ computable only if $S$ is
computable.  Here, we specialize to $\mathcal{Y}\subseteq\{A\in\Pi^0_1(X)\setconstr \dim A\leq 0\}$ with
$\delta_{\mathcal{Y}}\leq \delta_{\Pi^0_1}|^{\mathcal{Y}}$, and find the respective conditions of computability of the operations $N,M,B,S$
are in fact equivalent in this restricted case.  More specifically, we
consider also
$N_0:\subseteq\Pi^0_1(X)^2\times\mathcal{Y}\crsarr\Sigma^0_1(X)^2$ defined by
\[
N_0(A,B,Y):=\{(U,V)\setconstr
A\cap Y\subseteq U\wedge B\cap Y\subseteq V\wedge Y\cap U\cap V=\emptyset\wedge Y\subseteq U\cup V\}
\]
($\dom N_0=\{(A,B,Y)\setconstr A\cap B\cap Y=\emptyset\}$) and show
$N\eqvW N_0$ and $N_0\leqsW S$.  Here if $f:\subseteq W\crsarr X$ and
$g:\subseteq Y\crsarr Z$, we have
\emph{Weihrauch reducibility} and \emph{strong Weihrauch reducibility}
respectively defined by
\begin{align*}
f\leqW g &:\iff
(\exists\text{ computable }H,K:\subseteq\N^{\N}\to\N^{\N})
\left(\forall\text{ realizer }G\text{ of }g\right)\\
&\left( F:=H\compose\langle \id_{\N^{\N}},G\compose K\rangle
\text{ a realizer of }f \right),\\
f\leqsW g &:\iff
(\exists\text{ computable }H,K:\subseteq\N^{\N}\to\N^{\N})
\left(\forall\text{ realizer }G\text{ of }g\right)\\
&\left( F:=H\compose G\compose K\text{ a realizer of }f \right).
\end{align*}
(For more on these notions, the reader can refer e.g.~to
\cite{BGPWCpxtyPreprint}).
Recall $N:\subseteq\Pi^0_1(X)^2\times\mathcal{Y}\crsarr\Sigma^0_1(X)^2$ is defined by
$\dom N=\{(A,B,Y)\setconstr A\cap B=\emptyset\}$ and
$N(A,B,Y) := \{(U,V)\setconstr A\subseteq U\wedge B\subseteq V\wedge Y\subseteq U\disjunion V\}$.
\begin{prop}
If $\mathcal{Y}\subseteq\{A\in\Pi^0_1(X)\setconstr \dim A\leq 0\}$ with
$\delta_{\mathcal{Y}}\leq \delta_{\Pi^0_1}|^{\mathcal{Y}}$ then $N\leqW N_0$.
\end{prop}
\begin{proof}
First, we define computable $H,K:\subseteq\Bairespc\to\Bairespc$ such that
$F:=H\compose\langle\id_{\Bairespc},G\compose K\rangle$ realizes $N$ whenever $G$ is a realizer of $N_0$.
Namely, let $K:=\id_{\Bairespc}$ and if $I,J,L$ are respective computable realizers of
$\cap:\Pi^0_1(X)^2\to\Pi^0_1(X)$, $\delta_{\mathcal{Y}}\leq\delta_{\Pi^0_1(X)}|^{\mathcal{Y}}$ and
$\union:\Pi^0_1(X)^2\to\Pi^0_1(X)$, and $P$ a computable realizer of
$t_4$, we let $H:\subseteq\Bairespc\to\Bairespc$ be defined by
\[
H\langle p,q\rangle:= P\langle L\langle \pr_1\pr_1 p, I\langle J(\pr_2 p),\pr_2 q\rangle\rangle,
 L\langle \pr_2\pr_1 p,I\langle J(\pr_2 p),\pr_1 q\rangle\rangle\rangle.
\]
Clearly $H$ is computable, so we verify $F$ realizes $N$.  Namely, let $(A,B,Y)\in\dom N$,
$p\in [\delta_{\Pi^0_1}^2,\delta_{\mathcal{Y}}]^{-1}\{(A,B,Y)\}$,
$q:=G(p)\in (\delta_{\Sigma^0_1}^2)^{-1}N_0(A,B,Y)$, so $(U,V):=\delta_{\Sigma^0_1}^2(q)$ have
\begin{gather*}
A\union (Y\setminus V) =
(\delta_{\Pi^0_1}\compose L)\langle \pr_1\pr_1 p, I\langle J(\pr_2 p),\pr_2 q\rangle\rangle =: C
\onespace\text{ and }\\
B\union (Y\setminus U) =
(\delta_{\Pi^0_1}\compose L)\langle \pr_2\pr_1 p,I\langle J(\pr_2 p),\pr_1 q\rangle\rangle\rangle =: D,
\end{gather*}
in particular $H\langle p,q\rangle \in (\delta_{\Sigma^0_1}^2)^{-1}t_4(C,D)$.  Denoting $(U',V'):=
(\delta_{\Sigma^0_1}^2\compose H)\langle p,q\rangle$ we have $A\subseteq C\subseteq U'$, similarly
$B\subseteq V'$ and
we will see $Y\subseteq U'\disjunion V'$ (hence $(U',V')\in N(A,B,Y)$): namely $Y\subseteq U\union V$ with
$Y\cap U=Y\setminus V\subseteq C\subseteq U'$, similarly $Y\cap V\subseteq V'$ (and already
$U'\cap V'=\emptyset$ by definition of $t_4$).
\end{proof}

Conversely, to show $N_0\leqsW N$, if $G$ is a realizer of
$N$, $(A,B,Y)\in\dom N_0$ and $I$, $J$ are computable realizers of
$\cap:\Pi^0_1(X)^2\to\Pi^0_1(X)$ and
$\delta_{\mathcal{Y}}\leq\delta_{\Pi^0_1(X)}|^{\mathcal{Y}}$, then we
have $(A\cap Y)\cap(B\cap Y)=A\cap B\cap Y=\emptyset$ so
$(A\cap Y,B\cap Y,Y)\in\dom N$, and we can consider computing some
$(U,V)\in N(A\cap Y,B\cap Y,Y)$.  Then
\(A\cap Y\subseteq U\wedge B\cap Y\subseteq V\wedge
Y\subseteq U\disjunion V\) implies $U\cap V\cap Y=\emptyset$ and
$(U,V)\in N_0(A,B,Y)$.
On the other hand, we will compare $N_0$ with $S$, defined as follows:
$S:\subseteq\Sigma^0_1(X)^{\N}\times\mathcal{Y}\crsarr\Sigma^0_1(X)^{\N}$,
$\dom S=\{((U_i)_i,Y)\setconstr\bigcup_i U_i\supseteq Y\}$,
$S((V_i)_i,Y):=\{(W_i)_i\setconstr (\forall i)W_i\subseteq V_i\wedge \bigcup_i W_i\supseteq Y
\text{ and $(W_i)_i$ pairwise disjoint}\}$.
\begin{prop}
$N_0\leqsW S$ holds for any represented class $\mathcal{Y}\subseteq\{A\subseteq X\setconstr \dim A\leq 0\}$.
\end{prop}
\begin{proof}
Let $K(p):=\langle\langle \pr_2\pr_1 p,\pr_1\pr_1 p,0^{\ordomega},0^{\ordomega},\dots\rangle,\pr_2 p\rangle$
and $H\langle q^{(0)},q^{(1)},\dots\rangle:=\langle q^{(0)},q^{(1)}\rangle$; then any realizer $G$ of $S$
gives rise to a realizer $F:=H\compose G\compose K$ of $N_0$.  Namely, if $(A,B,Y)\in\dom N_0$,
$p\in [\delta_{\Pi^0_1}^2,\delta_{\mathcal{Y}}]^{-1}\{(A,B,Y)\}$ and
$q:=\langle q^{(0)},\dots\rangle:=(G\compose K)(p)$ then
$q\in(\delta_{\Sigma^0_1}^{\ordomega})^{-1}S(X\setminus B,X\setminus A,\emptyset,\emptyset,\dots;Y)$.
Denoting $(W_i)_i=\delta_{\Sigma^0_1}^{\ordomega}(q)$ by definition of $S$ we have
$W_0\cap B=\emptyset=W_1\cap A\wedge Y\subseteq W_0\disjunion W_1$.  But then $A\cap Y\subseteq W_0$,
$B\cap Y\subseteq W_1$ and $Y\cap W_0\cap W_1=\emptyset$, so $(W_0,W_1)\in N_0(A,B,Y)$, and
$(W_0,W_1)=(\delta_{\Sigma^0_1}^2\compose F)(p)$.
\end{proof}

\section{Concluding Remarks}\label{sec:concl}
The results of this paper, in particular the reductions of Section \ref{sec:zdsubsets} along with
Theorem \ref{thm:a} and Proposition \ref{prn:aconverse}, provide some details on a robust notion of
effective zero-dimensionality for classes $\mathcal{Y}$ of subsets of a computable metric space $X$, extending
the notion of effective zero-dimensionality for $X$ that was examined in \cite[\S 5]{zerodpaper}.  It would be
good to find precisely the necessary compactness and other assumptions for Proposition \ref{prn:aconverse}, and
to clarify in terms of Weihrauch reducibility the results of Theorem \ref{thm:a} and
Proposition \ref{prn:aconverse}, at least for classes of closed subsets.

On the other hand, in Theorem \ref{thm:b} we have followed the proof of the Dugundji Extension Theorem outlined
in \cite[Hint to Prob 4.5.20(a)]{Engelking} to compute a more general form of retraction, seemingly requiring a
valued field structure (and use of convexity in ultrametric spaces, along with the proof technique of
Theorem \ref{thm:a}) in the process.  It is interesting (though perhaps more complicated) to ask whether some
weaker algebraic conditions would suffice, e.g.~what the `general form' of retractions should be in
zero-dimensional topological groups?  Here note, regarding valued fields, we
have stated Theorem \ref{thm:b} partly for a zero-dimensional computable metric space $X$ and partly for a
$p$-adic field $K=\Omega_p$,
with the following well-known embedding theorem in mind:
\begin{thm}\label{thm:zdembcor}(see \cite[\S A.10, Corollary]{Schikhof})
Any separable ultrametric space $X$ has an isometric embedding into some separable nonarchimedean valued field
$K$.
\end{thm}
An effective version of this result may be of interest, in part to simplify the assumptions of
Theorem \ref{thm:b}, and in any case the notions appearing in \cite[\S 18--21]{Schikhof} (on ultrametric spaces;
see also the references of \cite{KKSNonarchDET}) deserve further attention from a computable analysis viewpoint.
As regards
$p$-adic fields or slightly different number-theoretic settings, there are very many concrete tools (and topics
of analysis) in the remainder of \cite{Schikhof} whose relation to computable analysis (and possibly
effective zero-dimensionality) may be worth investigating, though our knowledge is limited.  To mention two sections close to the topics already discussed, \cite[\S 26, \S 76]{Schikhof} respectively concern locally
constant ($K$-valued) functions \& differentiability, and extension of functions (including isometries).

Also, \cite[Prob 6.3.2(f)]{Engelking} supplies conditions on a zero-dimensional computable metric space $X$
under which the topology is induced by a linear order $(<)\subseteq X^2$.

Finally, computability of $Q'$ in the proof of Proposition \ref{prn:dugundji} has shown how a
lower semicontinuous function $f:X\setminus A\to (0,1]$ satisfying certain nonuniform upper bounds (dependent
on $\epsilon>0$) computably gives rise to a Dugundji system for $A\in\mathcal{A}(X)\setminus\{\emptyset,X\}$
(with coefficient $1+\epsilon$), without involving zero-dimensionality.  Here it is of interest to
see how such an $f$ may be chosen to be lower semicomputable without necessarily being continuous.  To provide
context, some results on lower semicontinuous functions and references to applications appear in
\cite{WZSemiCts00}; here we leave unanswered the question of applications of the general Dugundji system
construction, and instead remark on two related questions around lower semicontinuity:
\begin{enumerate}
\item
If $(X,\mathcal{T}_X,\alpha)$ is an effective topological space, each characteristic function
$\charf{U}:X\to\R$ ($U\in\Sigma^0_1(X)$) is $(\delta_X,\rho_{<})$-continuous, and
$u:\N\times X\to\R,(i,x)\mapsto\charf{\alpha(i)}(x)$ is $(\id_{\N},\delta_X;\rho_{<})$-computable.
More generally,
\[s:\subseteq \N\times\Bairespc\times\Bairespc^{\N}\to\Cont_{\Sigma^0_1(X)}(X,\R_{<}),
(N,p;p^{(0)},\dots)\mapsto \sum_{i<N}\rho_{<}(p^{(i)}).\charf{\alpha'(p_i)}|_{\delta_{\Sigma^0_1(X)}(p)}
\]
(with $\dom s=\{(N,p;p^{(0)},\dots)\setconstr (\forall i<N)p^{(i)}\in\rho_{<}^{-1}[0,\infty)\}$) is
$(\id_{\N},\id_{\Bairespc},\id_{\Bairespc}^{\ordomega};[\delta_X\to\rho_{<}]_{\Sigma^0_1(X)})$-computable.

One might consider
those $f\in\Cont_{\Sigma^0_1(X)}(X,\R_{<})$ which lie in the range of $s$; in the case that
$\img\alpha\subseteq\Delta^0_1(X)$ we know such functions are locally constant; for $K$-valued functions (in
particular, for rational-valued functions when $X$ is the set of $p$-adic integers in $K=\Omega_p$) that would
also imply $f'=0$ under suitable assumptions on $X$ and $\dom f$; see \cite[\S 26]{Schikhof}.

One can note further that a certain sequence of locally constant functions (from the $p$-adic integers to $K$) is used to define a basis of the $\Cont^1$ functions with zero derivative (with the same domain and codomain) in
\cite[Thm 68.1]{Schikhof}; we leave this to one side here and just discuss the locally constant functions.
Under one weak definition of monotone sequences, any continuous function $f:X\to K$ is the uniform
limit of a monotone sequence of locally constant $f_n:X\to\Q$ ($n\in\N$), cf.~\cite[Ex 87.H(i)]{Schikhof}.
It is of interest to examine whether this or another sense of monotone sequence (of locally constant functions)
gives rise to representations of $\Cont(X,K)$ analogous to the restriction
$[\delta_X\to\rho_{<}]|^{\Cont(X,\R)}$ in the real case.
%
%
%
%
\item
If $X$ is a complete computable metric space, $\mathcal{X}$ the Borel $\sigma$-algebra of $X$ and
$\mu:\mathcal{X}\to [0,1]$ is a probability measure, recall $\mu$ is a
\emph{computable probability measure on $X$} if
$\N^{*}\to\R_{<},w\mapsto \mu\left(\bigcup_{i<\absval{w}}\alpha(w_i)\right)$ is computable.
Note this is equivalent to $g_{\mu}:=\mu|_{\Sigma^0_1(X)}:\Sigma^0_1(X)\to\R_{<}$ being computable,
e.g.~\cite[Thm 4.2.1]{HRMses09}.  More generally, we can introduce a representation of the set
$\mathcal{M}(X)$ of finite measures $\mu:\mathcal{X}\to[0,\infty)$ by
\[
p\in\delta_{\mathcal{M}(X)}^{-1}\{\mu\}\iff p\in[\delta_{\Sigma^0_1(X)}\to\rho_{<}]^{-1}\{g_{\mu}\}.
\]
If the construction of a Dugundji system in (the proof of) Proposition \ref{prn:dugundji} can be adequately
generalised to nonmetrizable spaces like $Z=\Sigma^0_1(X)$ (with the topology from Remark \ref{rem:szoneadm}),
then a restriction of $g_{\mu}$ (for a probability measure $\mu\in\mathcal{M}(X)$) to a specified open subspace
of $Z$ may provide an example of a lower semicomputable function $f:\subseteq Z\to (0,1]$ with open domain.
(Note if $Y\in\Sigma^0_1(X)$ is a proper nonempty open subset then $A:=\Sigma^0_1(Y)$ is closed
and nonopen in $Z$).

Of course, in the setting of $p$-adic fields there are several analogues of measure and integration (both
$K$-valued and real-valued) appearing in \cite{Schikhof}, which may also be worth considering in this connection.
\end{enumerate}
\section*{Acknowledgement}
The author is very grateful for the useful comments of the anonymous referees, which have helped to improve the paper, inclusive of correcting some errors.
\bibliographystyle{alpha}
\bibliography{short5}
\end{document}